\documentclass[aps,pra,twocolumn,nofootinbib,superscriptaddress,groupedaddress]{revtex4}

\usepackage{graphicx}  % needed for figures
\usepackage{dcolumn}   % needed for some tables
\usepackage{bm}        % for math
\usepackage{amssymb}   % for math

\usepackage{hhline}
\usepackage{scalefnt}
\usepackage{amsmath}
\usepackage{color}
\usepackage{xr}
\usepackage{cases}

\usepackage{slashbox}
 \usepackage{enumitem}
 \usepackage{relsize}
\input{Qcircuit}
\def\I {\mathrm{i}}

\newcommand{\ket}[1]{\left| #1 \right>} % for Dirac bras
\newcommand{\bra}[1]{\left\langle #1 \right|} % for Dirac kets
\newcommand{\braket}[2]{\left\langle#1|#2\right\rangle}
\newcommand{\ketbra}[2]{|#1\rangle\!\langle#2|}
\newcommand{\proj}[1]{|#1\rangle\!\langle#1|}
\newcommand{\id}{\leavevmode\hbox{\small1\normalsize\kern-.33em1}}
\newcommand{\tr}{\mathrm{tr}}
\newcommand{\nt}{{\sc not}}
\newcommand{\cnot}{{\sc C-not}}
\newcommand{\cnots}{\cnot s}
\newcommand{\ot}{\otimes}
\newcommand{\cH}{\mathcal{H}}
\newcommand{\cO}{\mathcal{O}}

\newtheorem{thm}{Theorem}
\newtheorem{lem}{Lemma}
\newtheorem{cor}{Corollary}
\newtheorem{rmk}{Remark}
\newenvironment{proof}[1][Proof]{\noindent\textbf{#1.} }{\ \rule{0.5em}{0.5em}}

\begin{document}

\title{Quantum Circuits for Isometries}
\author{Raban~Iten} \affiliation{ETH Z\"urich, 8093 Z\"urich, Switzerland (itenr@student.ethz.ch)}  
\author{Roger~Colbeck} \affiliation{Department of Mathematics, University of York, YO10 5DD, UK (roger.colbeck@york.ac.uk)}       
\author{Ivan~Kukuljan} \affiliation{University of Ljubljana, 1000 Ljubljana, Slovenia}         
\author{Jonathan~Home} \affiliation{Institute for Quantum Electronics, ETH Z\"urich, Otto-Stern-Weg 1, 8093 Z\"urich, Switzerland}         
\author{Matthias~Christandl} \affiliation{Department of Mathematical Sciences, University of Copenhagen,
  Universitetsparken 5, 2100 Copenhagen \O, Denmark}

\date{$9^{\text{th}}$ April 2020}

\begin{abstract}
  We consider the decomposition of arbitrary isometries into a sequence of single-qubit and Controlled-{\nt} (\cnot) gates.  In many experimental architectures, the \cnot{} gate is relatively `expensive' and hence we aim to keep the number of these as low as possible. We derive a theoretical lower bound on the number of \cnot{} gates required to decompose an arbitrary isometry from $m$ to $n$ qubits, and give three explicit gate decompositions that achieve this bound up to a factor of about two in the leading order.  We also perform some bespoke optimizations for certain cases where $m$ and $n$ are small.  In addition, we show how to apply our result for isometries to give decomposition schemes for arbitrary quantum operations and POVMs via Stinespring's theorem. These results will have an impact on experimental efforts to build a quantum computer, enabling them to go further with the same resources.
\end{abstract}

%\pacs{03.67.Lx, 02.10.Yn, 03.67.Ac, 42.50.Dv} 
\maketitle

\begin{table*}[!t] 
\renewcommand{\arraystretch}{1.5}
\caption{Lowest known upper bounds and highest known lower bounds
  on the number of \cnot{} gates  required to decompose $m$ to
  $n$ isometries for  large $n$. For simplicity, all the counts are depicted to leading order.  As is to be expected, the number of
  required \cnot{} gates increases with $m$ (i.e., when fewer of the
  input qubits start in a fixed state).}
\label{tab:Results}
\centering
\begin{ruledtabular}
\begin{tabular}{llllll}

$m$& Lower Bound [LB] &Upper Bound [UB] &UB/LB &References for Upper bound  \\ \hline
$m=0$ (SP)&$\frac{1}{2}2^{n}$~\cite{3} &$\frac{23}{24}2^{n}$&$\simeq 1.9$&\cite{3}  ($n$ even), Rmk.~\ref{spodd} ($n$ odd)\\
$1\leqslant m \leqslant n-2 $ &$\frac{1}{2}2^{n+m}- 4^{m-1}$   &$2^{n+m}-\frac{1}{24}2^n$&$< 2.3$\footnote{If $ 1\leqslant m\leqslant n-5$ we have UB/LB$\,\lesssim 2$ (for large enough $n$).}&Eq.~(\ref{eq:Iso_CNOT}), (Theorem~\ref{thm_Knill})\footnote{In the case $5\leqslant m\leqslant n-2$ and even $n$, Theorem~\ref{thm_Knill} achieves a slightly lower \cnot{} count of  $\frac{23}{24}(2^{n+m}+2^n)$ to leading order.}\\
$m=n-1$ 	& $ \frac{3}{16} 4^n $& $ \frac{23}{64} 4^{n}$&$\simeq 1.9$&Eq.~(\ref{eq:CSD_count_2})\\
$m=n$ (Unitary)  & $\frac{1}{4}4^{n}$~\cite{unitary_lowerb1, unitary_lowerb2} &$ \frac{23}{48}4^{n}$&$\simeq 1.9$	&	\cite{2}	 \\ [0.04cm]										   
\end{tabular}
\end{ruledtabular}
\end{table*}

\section{Introduction}
Quantum computers would allow us to speed up several
important computations including search~\cite{grover1, grover2}, quantum
simulation~\cite{feynman} and factoring~\cite{shor}.  The ability
to do the latter would render RSA~\cite{RSA}, a widespread
cryptographic protocol, unfit for purpose.  However, constructing a
device capable of performing such computations is one of the biggest
challenges facing the field, and many candidate platforms remain in
their infancy, operating only with a few qubits at best.

In spite of this, the theory of quantum computation is quite advanced.
At an abstract level, a quantum computation corresponds to a unitary
operation, and a universal quantum computer should be able to perform
arbitrary unitary operations (each to very high precision). Rather
than having a different component for each unitary operation, it is
convenient to break down such operations in terms of a small family of
simple-to-perform gates.  This is the aim of the circuit model of
quantum computation, which mirrors an analogous model for classical
computation, in which an arbitrary computation can be decomposed in
terms of (for example) {\sc not}, {\sc and}, {\sc or} and \cnot{}
gates.  In the quantum case, several examples of universal gate
libraries are known (see for example~\cite{Buch}).  In this work we
focus on one involving arbitrary single-qubit operations and \cnot{}
gates. This gate set is universal for quantum computation in the sense
that an arbitrary $n$-qubit unitary can be decomposed in terms of
these gates alone~\cite{5} and is particularly well-suited to certain
architectures in which these operations are relatively straightforward
to implement. Of these operations, \cnot{} is often the most difficult
to perform since in all experimental architectures it involves
connecting the qubits using an additional degree of
freedom~\cite{WinelandBlatt2008,SuperconductingReview}. This provides
additional channels for the introduction of decoherence. The mediated
interaction also typically requires longer gate times, increasing
susceptibility to direct qubit decoherence. As an example, the
current lowest infidelities achieved experimentally are $< 10^{-6}$
for single-qubit gates~\cite{HartyPRL} and $\sim\!10^{-3}$ for two
qubit gates~\cite{Ballance}. Taking this as our motivation, we use the
number of \cnot{} gates required in a decomposition as a measure of
the complexity of a gate sequence and we consider circuits that
minimize the number of such gates.

This task has been previously considered both for arbitrary unitary
operations and for state preparation (see for example~\cite{2, 3} and
references therein).  In~\cite{2}, a decomposition scheme was found
for an arbitrary unitary operation on $n$ qubits that requires
$\frac{23}{48}4^n$ \cnots{} to leading order, approximately twice as
many as the best known lower bound~\cite{unitary_lowerb1,
  unitary_lowerb2}. Similarly, in order to prepare a state of $n$
qubits (starting from the state $\ket{0}^{\otimes n}$), the best known
construction requires $\frac{23}{24}2^n$ \cnots{} to leading order if
$n$ is even~\cite{3}, and $2^n$ to leading order if $n$ is
odd~\cite{10}, which is again approximately twice the best known lower
bound~\cite{3}.

State preparation and arbitrary unitaries are special cases of a wider
class of operations, \emph{isometries}.  An isometry is an
inner-product preserving transformation that maps between two Hilbert
spaces that in general have different dimensions.  Physically,
isometries can be thought of as the introduction of ancilla qubits in
a fixed state (conventionally $\ket{0}$) followed by a general unitary
on the system and ancilla qubits.  However, because its action only
has to be specified when the ancilla systems start in state $\ket{0}$,
there is a lot of freedom when constructing the general unitary.  This
freedom can be exploited to lower the number of \cnots{} needed with
respect to that of a general unitary.  In the special case where the
input and output spaces have the same dimensions, the isometry is a
unitary operation, while state preparation corresponds to an isometry
from a (trivial) one-dimensional space to that of the required output.
In this manuscript we consider the problem of synthesis of general
isometries from $m$ qubits to $n\geqslant m$ qubits.

\begin{table*}[!t] 
%% increase table row spacing, adjust to taste
\renewcommand{\arraystretch}{1.5}
\caption{Overview of the number of \cnot{} gates required to decompose $m$ to
  $n$ isometries %for  large $n$ 
  using different decomposition schemes (NB: for small $n$ we have
  done some additional optimizations|see Table~\ref{Tab:small_cases}). Abbreviations used: $^a$Column-by-column
  decomposition of an isometry; $^b$Decomposition of an isometry using
  the Cosine-Sine Decomposition.}
\label{tab:Results2}
\centering
\begin{ruledtabular}
\begin{tabular}{lllll}

Method & \cnot{} count for an $m$ to $n$ isometry & References  \\ \hline
Knill (optimized) & $\frac{23}{24}(2^{m+n}+2^n)+\mathcal O\left(n^2 \right)2^m$ \text{ if $n$ is even} &Theorem~\ref{thm_Knill}\\
		   &$ \frac{115}{96}(2^{m+n}+2^n)+\mathcal O\left(n^2 \right)2^m$ \text{ if $n$ is odd}&Theorem~\ref{thm_Knill}\\
CCD$^a$ &$2^{m+n}-\frac{1}{24}2^n+  \mathcal O\left(n^2 \right)2^m$&Eq.~(\ref{eq:Iso_CNOT})\\
CSD$^b$ & $\frac{23}{144}\left(4^m+2\cdot4^n\right)+  \mathcal O\left(m\right)$&Eq.~(\ref{eq:CSD_count_2}) \\ [0.04cm]										   
\end{tabular}
\end{ruledtabular}
\end{table*}

This task was first considered by Knill~\cite{Knill}, whose
decomposition scheme is based on a decomposition scheme for state
preparation (and uses such a scheme as a black box). His decomposition
scheme together with the state preparation scheme of~\cite{10}
(or~\cite{3}) leads directly (without any optimizations) to an
decomposition of $m$ to $n$ isometries requiring about $2 \cdot
2^{m+n}$ \cnots{} to leading order.  However, this can be modified
(together with the decomposition scheme for state preparation
described in~\cite{3}) to achieve $2^{m+n}+2^n$ to leading order,
which is our first decomposition scheme.

We also introduce two others.  Our second scheme is a column-by-column
decomposition of an isometry that requires about $2^{m+n}$ \cnot{}
gates to leading order. This decomposition also performs well for
cases where $m$ and $n$ are small.  For our final scheme, we adapt the
decomposition of arbitrary unitaries~\cite{2} to isometries, leading
to a \cnot{} count of about $0.16\cdot \left(4^m+2\cdot4^n\right)$ to
leading order. 

To compare the quality of our schemes we give a theoretical lower
bound on the number of \cnot{} gates required to decompose arbitrary
isometries. These results are summarized in Tables~\ref{tab:Results}
and~\ref{tab:Results2}. As shown in Table~\ref{tab:Results}, for large
enough $n$, in the worst case our decomposition scheme uses roughly
$2.3$ times the number of \cnots{} required by the lower bound (the
worst-case being an $n-2$ to $n$ isometry). This is comparable to the
factor of $1.9$ already known in the special cases of state
preparation and of arbitrary unitary operations.

In addition, we optimize the \cnot{} counts for $m$ to $n\leqslant 4$
isometries in Appendix~\ref{small_cases_appendix} (see
Table~\ref{Tab:small_cases} for a summary). These are most likely to
be of practical relevance for experiments performed in the near
future.

\begin{table}[!t] 
%% increase table row spacing, adjust to taste
\renewcommand{\arraystretch}{1.1}
\caption{Smallest known achievable \cnot{} counts for $m$ to $2\leqslant n\leqslant 4$ isometries. The counts for $n=m$ are as in~\cite{2}. The counts for state preparation ($m=0$) on two and three qubits are taken from~\cite{SP_three_qubits1}, and the count for state preparation on four qubits follows from the decomposition scheme described in Appendix~\ref{opt_SP_app}. The remaining cases are discussed in Appendix~\ref{small_cases_appendix}.  Note that the \cnot{} counts grow very fast. For example, any unitary on 10 qubits can be performed using about 500000 \cnot{} gates.}
\label{Tab:small_cases}
\centering
\begin{ruledtabular}
\begin{tabular}{ c|ccccccccccc}
\backslashbox{{\it n}}{{\it m}} &$0$&$1$&$2$&$3$&$4$  \\ \hline
$2$&$1$&$2$ &$3 $&$- $&$- $\\ 
$3$&$3$&$9$ &$14$&$20$&$- $\\ 
$4$&$8$&$22$&$54$&$73$&$100$\\ 
\end{tabular}
\end{ruledtabular}
\end{table}

The \cnot{} counts in Table~\ref{tab:Results},
Table~\ref{tab:Results2} and Table~\ref{Tab:small_cases} can be
directly used to upper bound the total number of gates needed for the
decomposition. Since each \cnot{} gate can introduce at most two
single-qubit gates into a quantum circuit without redundancy (cf.\
Section~\ref{sec:lower_bound} for similar arguments\footnote{Note
  that we count arbitrary single-qubit gates here (rather than gates
  that rotate about a fixed axis).}), the number of single-qubit gates
required for an isometry can be bounded by doubling the counts given
in the two tables and adding $n$, the number of qubits in
question.\\

Although we have ranked the decompositions in terms of gate counts
above, there may be other features of a given decomposition scheme
that make it preferable to another which may depend on the physical
setup. It is also interesting to note that our decomposition schemes
use others in a black box fashion (cf.\ Section~\ref{sec:results} for
more details), e.g., the decomposition scheme of Knill uses a scheme
for state preparation as a black box. An improvement in the
decomposition of the black box would therefore directly improve the
corresponding decomposition for an isometry, potentially altering the
ordering in terms of gate counts.

\section{Background information and notation}
We work in the circuit model of quantum computation in which the
fundamental information carriers are qubits.  A computational basis
state of the $2^n$-dimensional Hilbert space
$\mathcal{H}_{n}=\mathcal{H}_{1}^{\otimes n}$ of an $n$ qubit register
can be written as $\ket{b_{n-1}} \otimes \ket{b_{n-2}}\otimes \dots
\otimes \ket{b_{0}} $ or, in short notation, as $ \ket{b_{n-1}
  {b_{n-2}} \dots{b_{0}}}$, where $b_{i} \in \{0,1 \}$. To abbreviate
further we write $ \ket{b_{n-1} {b_{n-2}} \dots
  {b_{0}}}=\ket{\sum_{i=0}^{n-1}b_{i}2^{i}}_n$, i.e., we interpret the
bit string $b_{n-1} {b_{n-2}}\dots{b_{0}}$ as a binary number. If
$n=1$ we omit the subindex.  Thus,
$\ket{1}_3=\ket{001}=\ket{0}\ot\ket{0}\ot\ket{1}$, for example.

In the circuit model of quantum computation,
information carried in qubit wires is modified by quantum gates, which
correspond mathematically to unitary operations.  In particular, we
will use the following single-qubit gates:
\begin{eqnarray} \label{eq4}
	R_{x}(\theta)&=&\left(\begin{array}{cc} \cos [ \theta/2 ] &- \I
            \sin [\theta/2] \\  - \I \sin [\theta/2] & \cos
            [\theta/2]  \end{array}\right);\\
\label{eq5}
	R_{y}(\theta)&=&\left(\begin{array}{cc} \cos [\theta/2] & -\sin [\theta/2] \\   \sin [\theta/2] & \cos [\theta/2]  \end{array}\right);\\
\label{eq6}
	R_{z}(\theta)&=&\left(\begin{array}{cc} e^{-\I \theta/2 } & 0
            \\   0 & e^{\I \theta/2} \end{array}\right)\, ,
\end{eqnarray}
which correspond to rotations by angle $\theta$ about the $x$-, $y$-
and $z$-axes of the Bloch sphere.  One important special case is the
\nt{} gate, $\sigma_x=\I R_x(\pi)$ in terms of which the \cnot{} gate
can be written as $\ketbra{0}{0}\otimes I+\ketbra{1}{1} \otimes
\sigma_x$.
\begin{lem}[ZYZ decomposition]\label{ZYZ} 
  For every unitary operation $U$ acting on a single qubit, there
  exist real numbers $\alpha,\beta,\gamma$ and $ \delta$ such that
\begin{equation} \label{eq7}
	U=e^{\I\alpha}R_{z}(\beta)R_{y}(\gamma)R_{z}(\delta).
\end{equation}
\end{lem}

A proof of this decomposition can be found in~\cite{Buch}. Note that
(by symmetry) Lemma~\ref{ZYZ} holds for any two orthogonal rotation
axes. Lemma~\ref{ZYZ} shows that a single-qubit gate can be specified
by three real parameters neglecting the (physically insignificant)
global phase $e^{\I\alpha}$.  This is analogous to the description of
a rotation in 3-dimensions being parameterized in terms of three
\emph{Euler angles}, here $\beta$, $\gamma$ and $\delta$.
 
It is convenient to represent quantum circuits diagrammatically.  Each
qubit is represented by a wire and gates are shown using a variety of
symbols. Conventionally time flows from left to right. We will use the
concept of circuit topologies, as in~\cite{unitary_lowerb1,
  unitary_lowerb2}, throughout this paper. A general circuit topology
corresponds to a set of quantum circuits that have a particular
structure, but in which some gates may be free or have free
parameters. For example, Lemma~\ref{ZYZ} can be expressed as an
equivalence of two circuit topologies.
 
\[
\Qcircuit @C=1.0em @R=.46em {
  &\gate{U}&\qw&=&&\gate{R_z}&\gate{R_y}&\gate{R_z}&\qw }
\]

The general meaning of a circuit topology equivalence is the
following: for all possible values of the (free) parameters of the
circuit topology on the left hand side there exist values for the
parameters of the circuit topology on the right hand side such that
the two sides perform the same operation (up to a global phase). For
example, each of the $R_z$ gates in the above circuit represents a
$z$-rotation gate with unspecified angle. If we use symbols for
certain gates that have not been introduced before, they are
considered to be arbitrary quantum gates (these will often be denoted
by $U$).  If the same symbol is used as a placeholder for more than
one quantum gate, we mean that all gates are of this form, but the
gates themselves don't have to be identical (as in the previous
example where although $R_z$ appears twice on the right hand side,
each instance can have a different rotation angle).

\section{Lower bound} \label{sec:lower_bound}
First we derive a theoretical lower bound on the number of \cnot{}
gates required to decompose an isometry. For this purpose we use a
similar argument as that used to derive theoretical lower bounds for
general quantum gates~\cite{unitary_lowerb1, unitary_lowerb2} or for
state preparation~\cite{3}. Let $m$ and $n$ be natural numbers with $n
\geqslant 2$ and $m \leqslant n$. An $m$ to $n$ isometry can be
represented by a $2^n\times2^m$ complex matrix satisfying
$V^{\dagger}V=I_{2^m\times 2^m}$. Therefore such an isometry is
described by $2^{n+m+1}-2^{2m}-1$ real parameters, where the $-1$
accounts for the physically negligible global phase.

We can think of this isometry in terms of a unitary
operation on $n$ qubits, $n-m$ of which always start in a fixed state,
which we take to be $\ket{0}$\footnote{Note that additional ancilla qubits will not affect the lower bound. This can be seen by using the same arguments that we use in the derivation of the lower bound for quantum channels (see Section~\ref{sec:ACO}).}. Without any \cnots{}, all we can do is
apply single-qubit unitaries individually to each of these $n$ qubits.
Each such unitary introduces at most 3 parameters
(cf.~Lemma~\ref{ZYZ}). However, for the qubits that start in state
$\ket{0}$, only two parameters are introduced, since a qubit state is
fully specified by two real parameters. In order to introduce further
parameters, \cnot{} gates are required.

One might expect each \cnot{} gate to allow the introduction of six
real parameters by placing arbitrary single-qubit rotations after the
control and target. However, since $R_z$ gates commute with control
qubits, and $R_x$ gates with target qubits, we can introduce at most
four parameters for each additional \cnot{} gate~\cite{unitary_lowerb1,
  unitary_lowerb2}.  In essence we are using the following circuit identity

\[
\Qcircuit @C=0.8em @R=.46em {
 & \ctrl{2} &\gate{R_z}&\gate{R_y}&\gate{R_z}&\qw&&& &\gate{R_z}& \ctrl{2} &\gate{R_y}&\gate{R_z}&\qw   \\
& &&&&&&	 =&& \\
 &\targ &\gate{R_x}&\gate{R_y}&\gate{R_x}&\qw& &&  &\gate{R_x}&\targ &\gate{R_y}&\gate{R_x}&\qw	 \\
}
\]
which implies
\begin{equation}\label{comm_to_left}
\Qcircuit @C=0.8em @R=.46em {
&\gate{U} & \ctrl{2} &\gate{U}&\qw&&& &\gate{U}& \ctrl{2} &\gate{R_y}&\gate{R_z}&\qw   \\
& &&&&&	 =&& \\
&\gate{U} &\targ &\gate{U}&\qw& &&  &\gate{U}&\targ &\gate{R_y}&\gate{R_x}&\qw	 \\
}
\end{equation}

We conclude, that we can introduce at most $3m+2(n-m)+4r$ real
parameters using $r$ \cnot{} gates.

In order to be a valid circuit topology, i.e., one that can generate
every $m$ to $n$ isometry by an appropriate choice of its parameters,
the number of parameters introduced into the circuit by the
single-qubit rotations must exceed the number of parameters required
to specify an arbitrary $m$ to $n$ isometry.  Thus, the number of
\cnots{} required for such a circuit topology,
$N_{\mathrm{iso}}(m,n)$, must satisfy
$3m+2(n-m)+4N_{\mathrm{iso}}(m,n)\geqslant 2^{n+m+1}-2^{2m}-1$. From
this we obtain the following lower bound

\begin{equation}\label{lb}
N_{\mathrm{iso}}(m,n)\geqslant \frac{1}{4}\left(2^{n+m+1}-2^{2m}-2n-m-1\right).
\end{equation}

We remark that we can rephrase our result (by similar arguments as
used in~\cite{unitary_lowerb1, unitary_lowerb2}) as follows: almost
every $m$ to $n$ isometry cannot be decomposed into a quantum circuit
(comprising single-qubit unitaries and \cnots{}) with fewer than
$\lceil \frac{1}{4}\left(2^{n+m+1}-2^{2m}-2n-m-1\right) \rceil$
\cnot{} gates.  It is worth saying that the set of measure zero that
is excluded from this statement contains several interesting
isometries, for example that required for Shor's algorithm~\cite{shor}.
This lower bound provides a limitation on a \emph{universal} quantum
computer, rather than one tailored to a specific task.

\section{Decomposition schemes for isometries} \label{sec_dec_schemes}
Any isometry, $V$, from $m$ qubits to $n$ qubits can be described by a
$2^n\times2^m$ matrix.  This can instead be represented by a
$2^n\times2^n$ unitary matrix, $U$, by writing $V=U I_{2^n\times2^m}$,
where $I_{2^n\times2^m}$ denotes the first $2^m$ columns of the
$2^n\times2^n$ identity matrix.  Note that $U$ is not unique (unless
$m=n$).  Our aim is to find a decomposition of a quantum gate of the
form $U$ in terms of \cnots{} and single-qubit gates. We describe
three constructive decomposition schemes for arbitrary
isometries. This section focuses on the ideas behind these
decomposition schemes; the full technical details can be found in
Appendix~\ref{sec:TD}.  It is also worth noting that the proof of each
of these schemes can be seen as an alternative way to
prove the universality of the gate library containing single-qubit and
\cnot{} gates~\cite{5}.

\subsection{Notation for controlled gates}  \label{sec:controlled gates} 

We use $l$-qubit-$C^{\textnormal{u}}_{k}(U)$ to denote a gate that
performs a different $l$-qubit unitary for each possible state of $k$
\emph{control} qubits, where $U$ is a placeholder for a size $2^k$ set
of $2^l$-dimensional unitary operations.  We call an operation of this
type a uniformly controlled gate (UCG). These are also referred to as
``multiplexed gates'' by some authors, e.g.~\cite{2}.  If $l=1$ we
abbreviate the notation to $C^{\textnormal{u}}_{k}(U)$.  If we write
$R_x$, $R_y$ or $R_z$ instead of $U$, we mean that all the $2^k$
single-qubit gates that determine the UCG are of the form of the
corresponding rotation gate.

In order to write such gates out more precisely, we split the Hilbert
space of $n$ qubits into a $2^k$-dimensional space corresponding to
the control-qubits, a $2^l$-dimensional space corresponding to the
target-qubits and a $2^f$-dimensional space, where $f:=(n-l-k)$,
corresponds to the free qubits, i.e., the qubits we neither control
nor act on: $\mathcal{H}_{n}=\mathcal{H}_{k} \otimes \mathcal{H}_{l}
\otimes \mathcal{H}_{f}$.
 If $F$ is an $l$-qubit-$C^{\textnormal{u} }_{k}(U)$ gate, then it
 acts according to
\begin{equation} \label{eq12}
	F \left(\ket{i_{1}}_{k} \otimes \ket{i_{2}}_{l} \otimes   \ket{i_{3}}_{f}\right) \!  =\!\ket{i_{1}}_{k}\otimes (U_{i_{1}}\ket{i_{2}}_{l})\otimes  \ket{i_{3}}_{f},
\end{equation}
where $i_1 \in \{0,\dots,2^k-1\}$, $ i_2 \in \{0,\dots,2^l-1\}$, $i_3
\in \{0,\dots,2^{f}-1\}$ and $U_{i_1}$ denotes the quantum gate acting
on the target qubits if the control qubits are in the state
$\ket{i_1}_{k}$.  If each member of the set ${U_{i_1}}$ apart from one
(call this one $U_j$) are equal to the identity operation, we drop the
word ``uniformly'' and call such an operation a $k$-controlled
$l$-qubit gate, denoted by $l$-qubit-$C_{k}(U_{j})$, or more generally
a multi-controlled gate (MCG). If $l=1$ and we want to emphasize the
total number $n$ of qubits of the system being considered, we add an
$n$ as a second subindex, i.e. $C_{k}(U)$ becomes $C_{k,n}(U)$.

By way of example, the following circuit diagram shows a
$2$-qubit-$C^{\textnormal{u}}_{2}(U)$, $C_{3}(U)$ (or $C_{3,4}(U)$)
and $C_{2}(U)$ (or $C_{2,4}(U)$) gate in this order (from left to
right).
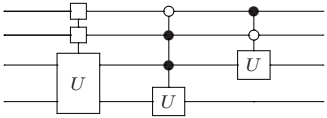
\[
\Qcircuit @C=2em @R=.3em {
&  \gate{} \qwx[1] \qw			    &\ctrlo{1} &\ctrl{1}	&\qw \\
& \gate{} \qwx[1] \qw 	 		    &\ctrl{1} &\ctrlo{1}	&\qw \\
& \multigate{1}{\mathsmaller{U}}  &  \ctrl{1} &    \gate{\mathsmaller{U}} &\qw  \\
& \ghost{\mathsmaller{U}}		    &    \gate{\mathsmaller{U}}   &\qw	&\qw  
}
\]
Note that the $C_{k}(U)$ notation does not specify which are the
control- and which are the target-qubits and whether we control on
$\ket{1}$ (filled circle)  or on $\ket{0}$ (unfilled circle); these must be made clear in the particular
context.

Each uniformly $k$-controlled gate can be decomposed into a
sequence of $2^{k}$ $k$-controlled gates, as should be clear from the
following example for the case $k=2$, $l=n-2$ and $n \geqslant3$.

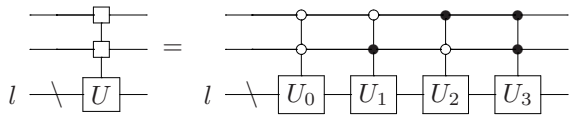
\[
\Qcircuit @C=1.0em @R=.7em {
&&\qw& \gate{} \qwx[1] \qw  	&\qw& &&&\qw    &    \ctrlo{1} \qw &  \ctrlo{1}  \qw   &  \ctrl{1}  \qw    &  \ctrl{1}  \qw        &	 \qw \\
&& \qw& \gate{} \qwx[1] \qw  	&\qw &  = &&&\qw   &  \ctrlo{1}  \qw&\ctrl{1}  \qw      &     \ctrlo{1}  \qw  &    \ctrl{1}  \qw    &   \qw\\
&\lstick{l}  &  {\backslash}\qw &\gate{U}  & \qw    		&& &  \lstick{l} &  {\backslash}    \qw   & \gate{U_{0}}  & \gate{U_{1}}
&\gate{U_{2}}     & \gate{U_{3}}   &    \qw \\
}
\]
The symbol ``${\backslash}$"  stands for a data bus of several (in this case $l$)  qubits. Note that the UCG above has block structure $U_{0} \oplus U_{1}  \oplus  U_{2}  \oplus  U_{3}$.

\begin{rmk} 
In Table~\ref{tab:counts_overview} of
Appendix~\ref{sec:counts_overview} we give an overview of  \cnot{}
counts for some special controlled gates that are used for
decompositions arising in this paper.
\end{rmk}

\subsection{Decomposition of isometries using the decomposition scheme of Knill}\label{sec:Knill}

In this section we combine the decomposition scheme for isometries of
Knill~\cite{Knill} and the state preparation scheme described
in~\cite{3}. The main result is as follows.

\begin{thm}\label{thm_Knill}
  Let $m$ and $n$ be natural numbers with $n \geqslant 5$ and $m
  \leqslant n$ and $V$ be an $m$ to $n$ isometry. There exists a
  decomposition of $V$ in terms of single-qubit gates and \cnots{} such
  that the number of \cnot{} gates required satisfies\footnote{A more precise count for this decomposition can be obtained by replacing
    $\cO(n^2)$ by $16n^2-60n+42$ (similarly in the two equations below).}
\begin{align}
N_{\mathrm{iso}}(m,n)\leqslant(2^{m}&+1)(N_U(\lfloor n/2 \rfloor)+N_U(\lceil n/2 \rceil)) \nonumber  \\
&+2^{m+1}N_{SP}(\lfloor n/2 \rfloor)+\cO\left(n^2\right)2^m, 
\label{eq:Iso_Knill_CNOT}
\end{align}
where $N_U(n)$ denotes the number of \cnot{} gates required for an arbitrary unitary on $n$ qubits. Using the best known \cnot{} counts for unitaries and state preparation (cf.\ Table~\ref{tab:Results} and Remark~\ref{spodd}) this leads to 
 \setlength{\arraycolsep}{0.0em}
\begin{eqnarray*}
N_{\mathrm{iso}}(m,n)&{} \leqslant {}& \tfrac{23}{24}(2^{m+n}+2^n)+\tfrac{23}{12}2^{m+\frac{n}{2}}-2^{m+\frac{n}{4}+2}\\&&\quad+\cO\left(n^2 \right)2^m \text{ if $n$ is even},  \\
N_{\mathrm{iso}}(m,n)&{} \leqslant {}& \tfrac{115}{96}(2^{m+n}+2^n)+\tfrac{23}{12}2^{m+\frac{n-1}{2}}-2^{m+\frac{n-1}{4}+2}\\&&\quad+\cO\left(n^2 \right)2^m \text{ if $n$ is odd}.  \nonumber
\end{eqnarray*}
\setlength{\arraycolsep}{5pt}
\end{thm}

\begin{rmk}
  For large $n$, the last three terms in~\eqref{eq:Iso_Knill_CNOT} are
  negligible. The leading order for this scheme is therefore derived
  from that of a unitary on $n/2$ qubits.
\end{rmk}

Consider a set of unitary operations $\{V_i\}_{i=0}^{2^m-1}$ such that
$V_i\ket{0}=V\ket{i}$, i.e., $V_i$ is a unitary for state preparation
on the state corresponding to the $i$th column of $V$.  In the proof
of Theorem~3.1 of~\cite{Knill} it is shown that 

%By the proof of Theorem~3.1 of~\cite{Knill}, for every isometry  $V$ from $m$ to $n$ qubits there exists a representation $U$ (i.e., a unitary matrix $U$ such that  $V=U I_{2^n\times2^m}$) that  can be decomposed into a sequence of  unitaries $V_i$ corresponding to state preparation, i.e.,  $V_i \ket{0}_n=\ket{\psi_i}$ for a desired state $\ket{\psi_i}$, and $C_{n-1}(P(\theta))$ gates as follows
\begin{equation}\label{eq:Knill_Dec}
U=V_{2^m-1}C_{n-1}(P({\theta_{2^{m-1}}}))V_{2^m-1}^{\dagger}\dots
V_{0}C_{n-1}(P({\theta_{0}}))V_{0}^{\dagger}\, ,
\end{equation}
where the gate $P(\theta):=e^{\I \theta }\proj{0}+\proj{1}$.  Consider
decomposing each $V_i$ using the (reverse of the) decomposition
scheme for state preparation described in~\cite{3}.
%, we decompose gates
%of the form $V_{i+1}^{\dagger}V_i$ simultaneously. The unitary $V_i$
%(and $V_{i+1}^{\dagger}$) can be decomposed using the decomposition
%scheme for state preparation (reversed) described in~\cite{3}. 
This leads to a circuit containing $2^m-1$ instances of the following
circuit diagram (shown in the case, where $n$ is even), each
corresponding to a unitary of the form $V_{i+1}^{\dagger}V_i$.

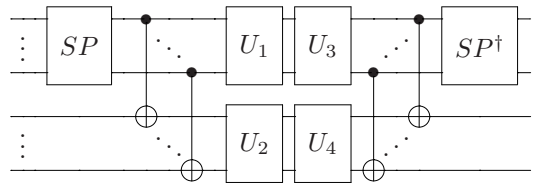
\[
\Qcircuit @C=0.45em @R=0.38em {
&&\qw&\qw&\multigate{3}{SP}&\qw&\ctrl{5}&\qw&\qw&\qw&\multigate{3}{U_1}&\multigate{3}{U_3}&\qw&\qw&\ctrl{5}&\multigate{3}{SP^{\dagger}}&\qw\\
&&\vdots&&&&&\ddots&&&&&&\reflectbox{$\ddots$}\\
&&&&&&&&&&&&&\\
&&\qw&\qw&\ghost{SP}&\qw&\qw&\qw&\ctrl{5}&\qw&\ghost{U_1}&\ghost{U_3}&\ctrl{5}&\qw&\qw&\ghost{SP^{\dagger}}&\qw\\
&&&&&&&&&&&&&\\
&&\qw&\qw&\qw&\qw&\targ&\qw&\qw&\qw&\multigate{3}{U_2}&\multigate{3}{U_4}&\qw&\qw&\targ&\qw&\qw\\
&&\vdots&&&&&\ddots&&&&&&\reflectbox{$\ddots$}\\
&&&&&&&&&&&&&\\
&&\qw&\qw&\qw&\qw&\qw&\qw&\targ&\qw&\ghost{U_2}&\ghost{U_4}&\targ&\qw&\qw&\qw&\qw
}
\]

We can merge the unitaries and define $\tilde{U}_1:=U_3U_1$ and $\tilde{U}_2:=U_4U_2$.

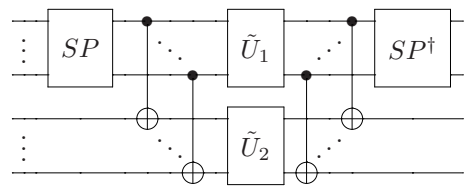
\[
\Qcircuit @C=0.45em @R=0.38em {
&&\qw&\qw&\multigate{3}{SP}&\qw&\ctrl{5}&\qw&\qw&\qw&\multigate{3}{\tilde{U}_1}&\qw&\qw&\ctrl{5}&\multigate{3}{SP^{\dagger}}&\qw\\
&&\vdots&&&&&\ddots&&&&&\reflectbox{$\ddots$}\\
&&&&&&&&&&&&&\\
&&\qw&\qw&\ghost{SP}&\qw&\qw&\qw&\ctrl{5}&\qw&\ghost{\tilde{U}_1}&\ctrl{5}&\qw&\qw&\ghost{SP^{\dagger}}&\qw\\
&&&&&&&&&&&&&\\
&&\qw&\qw&\qw&\qw&\targ&\qw&\qw&\qw&\multigate{3}{\tilde{U}_2}&\qw&\qw&\targ&\qw&\qw\\
&&\vdots&&&&&\ddots&&&&&\reflectbox{$\ddots$}\\
&&&&&&&&&&&&&\\
&&\qw&\qw&\qw&\qw&\qw&\qw&\targ&\qw&\ghost{\tilde{U}_2}&\targ&\qw&\qw&\qw&\qw
}
\]

We decompose all the terms of the form $V_{i+1}^{\dagger}V_i$ in
equation~(\ref{eq:Knill_Dec}) in this way. The gate $V_{2^m-1}$ and
$V_0^{\dagger}$ can also be decomposed using the (reversed)
decomposition scheme for state preparation described in~\cite{3}.  The
$C_{n-1}(P(\theta_i))$ gates are special cases of $C_{n-1}(U)$
gates. Hence, each can be decomposed into $16n^2-60n+42$ \cnot{} gates
(see Lemma~\ref{ThmC2}). This leads to the claimed \cnot{} count given
in equation~(\ref{eq:Iso_Knill_CNOT}).

\begin{figure*}[!t] 
\centering

 \renewcommand{\arraystretch}{0.5}
\begin{equation*}
 \mathsmaller{\ket{\psi_0^0}} = \begin{array}{c@{\!\!\!}l}
  \left[ \begin{array}[c]{ccccc} 
     \mathsmaller{\ast}\\
     \mathsmaller{\ast}\\
  \hline
    \mathsmaller{\ast}\\
      \mathsmaller{\ast}\\
    \hline
     \mathsmaller{\ast}\\
     \mathsmaller{\ast}\\
    \hline
    \mathsmaller{\ast}\\
      \mathsmaller{\ast}\\
    \hline
     \mathsmaller{\ast}\\
     \mathsmaller{\ast}\\
    \hline
    \mathsmaller{\ast} \\
      \mathsmaller{\ast}\\
    \hline
     \mathsmaller{\ast}\\
     \mathsmaller{\ast}\\
    \hline
    \mathsmaller{\ast}\\
      \mathsmaller{\ast}\\
  \end{array}  \right]
\end{array}
\phantom{a}
 \xrightarrow{ C_{3}^{\textnormal{u}}(U_{0,0}^{\textnormal{u}}) } 
 \begin{array}{c@{\!\!\!}l}
   \left[ \begin{array}[c]{ccccc} 
     \mathsmaller{\ast}\\
     \mathsmaller{0}\\
    \mathsmaller{\ast}\\
     \mathsmaller{0}\\
    \hline
     \mathsmaller{\ast}\\
    \mathsmaller{0}\\
    \mathsmaller{\ast}\\
     \mathsmaller{0}\\
    \hline
     \mathsmaller{\ast}\\
    \mathsmaller{0}\\
    \mathsmaller{\ast}\\
     \mathsmaller{0}\\
    \hline
     \mathsmaller{\ast}\\
    \mathsmaller{0}\\
    \mathsmaller{\ast}\\
     \mathsmaller{0}\\
\end{array}  \right]
\end{array}
\phantom{a}
 \xrightarrow{C_{2}^{\textnormal{u}}(U_{0,1}^{\textnormal{u}})}  
 \begin{array}{c@{\!\!\!}l}
   \left[ \begin{array}[c]{ccccc} 
     \mathsmaller{\ast}\\
    \mathsmaller{0}\\
   \mathsmaller{0}\\
     \mathsmaller{0}\\
     \mathsmaller{\ast}\\
   \mathsmaller{0}\\
  \mathsmaller{0}\\
     \mathsmaller{0}\\
    \hline
     \mathsmaller{\ast}\\
    \mathsmaller{0}\\
   \mathsmaller{0}\\
     \mathsmaller{0}\\
     \mathsmaller{\ast}\\
    \mathsmaller{0}\\
   \mathsmaller{0}\\
     \mathsmaller{0}\\
\end{array}  \right]
\end{array}
\phantom{a}
 \xrightarrow{C_{1}^{\textnormal{u}}(U_{0,2}^{\textnormal{u}})}  
 \begin{array}{c@{\!\!\!}l}
   \left[ \begin{array}[c]{ccccc} 
     \mathsmaller{\ast}\\
    \mathsmaller{0}\\
   \mathsmaller{0}\\
     \mathsmaller{0}\\
    \mathsmaller{0}\\
 \mathsmaller{0}\\
     \mathsmaller{0}\\
    \mathsmaller{0}\\
         \mathsmaller{\ast}\\
    \mathsmaller{0}\\
   \mathsmaller{0}\\
     \mathsmaller{0}\\
    \mathsmaller{0}\\
 \mathsmaller{0}\\
     \mathsmaller{0}\\
    \mathsmaller{0}\\
\end{array}  \right]
\end{array}
\phantom{a}
 \xrightarrow{C_{0}^{\textnormal{u}}(U_{0,3}^{\textnormal{u}})}  
\begin{array}{c@{\!\!\!}l}
  \left[ \begin{array}[c]{ccccc} 
      \mathsmaller{1}\\
   \mathsmaller{0}\\
     \mathsmaller{0}\\
    \mathsmaller{0}\\
 \mathsmaller{0}\\
     \mathsmaller{0}\\
    \mathsmaller{0}\\
       \mathsmaller{0}\\
    \mathsmaller{0}\\
   \mathsmaller{0}\\
     \mathsmaller{0}\\
    \mathsmaller{0}\\
 \mathsmaller{0}\\
     \mathsmaller{0}\\
    \mathsmaller{0}\\
        \mathsmaller{0}\\
         \end{array}  \right]
  \end{array}
\end{equation*}

\[
\Qcircuit @C=5em @R=.05em {
&& \gate{} \qwx[1] \qw 		&	 \gate{} \qwx[1] \qw 					&	\gate{} \qwx[1] \qw			&\gate{U^{\textnormal{u}}_{0,3}}		 &\qw			\\
&& \gate{} \qwx[1] \qw			&	 \gate{} \qwx[1] \qw  					&	\gate{U^{\textnormal{u}}_{0,2}}	 &\qw 					&\qw 		\\
&& \gate{} \qwx[1]  \qw 		&	\gate{U^{\textnormal{u}}_{0,1}} \qw&	\qw 						&\qw 					&\qw					\\
&& \gate{U^{\textnormal{u}}_{0,0}}  \qw &\qw						&	\qw						&\qw 					&\qw 						 \\
}
\]

\caption{Implementing the first column of an isometry $V$ from
  $m\geqslant0$ qubits to $n=4$ qubits. The action of $G_0$ on
  $\ket{{\psi_0^0}}:=V \ket{0}_m$ can be decomposed into operators
  $\{G_0^i\}_{i \in {0,1,2,3}}$, where
  $G_0^i:=C_{3-i}^{\textnormal{u}}(U_{0,i}^{\textnormal{u}})$.  The
  upper part shows how these gates successively zero the entries of
  the column, while the lower part gives the circuit representation.
  The inverse of this decomposition scheme was introduced in~\cite{10}
  for state preparation together with an efficient decomposition of
  the uniformly controlled gates $G_0^{i}$ into \cnots{} and
  single-qubit gates. The symbol ``$\ast$" denotes an arbitrary
  complex number.}
\label{fig:coloumn1}
\end{figure*}

 \renewcommand{\arraystretch}{1}

\subsection{Column-by-column decomposition}\label{sec:isounif}
In this section we introduce a circuit topology corresponding to a
column-by-column decomposition of an arbitrary isometry, i.e., we
decompose any isometry into single-qubit and \cnot{} gates proceeding
one column at a time.

\begin{thm}\label{thm_main}
  Let $m$ and $n$ be natural numbers with $n \geqslant 2$ and $m
  \leqslant n$ and $V$ be an $m$ to $n$ isometry. There exists a
  decomposition of $V$ in terms of single-qubit gates and \cnots{}
  such that the number of \cnot{} gates required satisfies
$$N_{\mathrm{iso}}(m,n)\leqslant 2^{m} (\Sigma_{s=0}^{n-1} N_{ \Delta C^{ \textnormal{u} }_{n-1-s}})+  \mathcal O\left(n^2 \right)2^m,$$
where $N_{\Delta C^{\textnormal{u}}_{n-1-s}}$ denotes the number of
\cnot{} gates required to decompose a $C^{\textnormal{u}}_{n-1-s}(U)$
gate up to a diagonal gate $\Delta$, i.e., to decompose the two gates
together, where the $C^{\textnormal{u}}_{n-1-s}(U)$ gate is determined
but we are free to choose the diagonal gate $\Delta$. Together with
the best known decomposition scheme for UCGs (up to diagonal
gates)~\cite{10} this leads to
$$N_{\mathrm{iso}}(m,n)\leqslant 2^{m+n}+  \mathcal O\left(n^2 \right)2^m.$$
\end{thm}

We defer a rigorous proof of the theorem to
Appendix~\ref{sec:app_isounif}, and instead use this section to
explain the main ideas behind the argument. Our proof is constructive,
and the exact \cnot{} count is given in equation~(\ref{eq:Iso_CNOT}).

As before, we represent the $m$ to $n$ isometry $V$ by a
$2^n\times2^n$ unitary matrix, here $G^{\dagger}$, by writing
$V=G^{\dagger}I_{2^n\times2^m}$.  Since a \cnot{} gate is inverse to
itself and the inverse of a single-qubit unitary is another
single-qubit unitary, searching for a decomposition scheme for
$G^{\dagger}$ is equivalent to searching for a decomposition of a
unitary operation $G$ satisfying $GV=I_{2^n\times2^m}$.

In essence, the idea is to find a sequence of unitary operations that
when applied to $V$ successively bring it closer to
$I_{2^n\times2^m}$.  We will do this in a column by column fashion,
first choosing a sequence of quantum gates, corresponding to a unitary
$G_0$ that gets the first column right, i.e., 
$G_0V\ket{0}_m=I_{2^n\times2^m}\ket{0}_m=\ket{0}_n$, we then use $G_1$
to get the second column right without affecting the first, i.e.,
$G_1G_0V\ket{1}_m=I_{2^n\times2^m}\ket{1}_m=\ket{1}_n$ and
$G_1G_0V\ket{0}_m=G_1\ket{0}_n=\ket{0}_n$, and so on (up to the $2^m$th column).
In other words, $G_k$ gets the $(k+1)$th column right and acts
trivially on the first $k$ columns of $I_{2^n\times2^m}$.

The gate $G_0$ can be decomposed into single-qubit and \cnot{} gates
by reversing a decomposition scheme for the preparation of a state
(applied to $V\ket{0}_m$).  It is natural to imagine
repeating this construction for each column in turn.  However, without
further modification, this procedure doesn't work since the action
required for the decomposition of later columns affects those that
have already been done. In other words, if we construct a unitary
$\tilde{G}_1$ again by reversing a decomposition scheme for state
preparation, we can obtain $\tilde{G}_1G_0V\ket{1}_m=\ket{1}_n$, but, in
general, $\tilde{G}_1G_0V\ket{0}_m \neq \ket{0}_n$. We therefore introduce a modified technique that takes this
into account while only slightly increasing the number of \cnot{}
gates needed over that required for state preparation on each
column. This technique develops an idea used for state preparation using
uniformly controlled gates~\cite{10}.

 \begin{lem} \label{rotate_to_basis_state} Let
   $\ket{\psi'}\in\mathcal{H}_{1}$ and define $r$ such that
   $\braket{\psi'}{\psi'}=r^2$. There exist $U_0,U_1 \in SU(2)$, such
   that
\begin{eqnarray} \label{eq9}
	U_0\ket{\psi'}&=&r \ket{0},\\
\label{eq10}
	U_1 \ket{\psi'}&=&r \ket{1} .
\end{eqnarray}
\end{lem}

\begin{proof} 
  Define $\ket{\psi}=\frac{1}{r}\ket{\psi'}$ and $\ket{\phi}=-\braket{\psi}{1}\ket{0}+\braket{\psi}{0}\ket{1}\in\cH_1$. Then
  $U_0=\ketbra{0}{\psi}+\ketbra{1}{\phi}$ is unitary with $\det U_0=1$ and
  obeys equation~\eqref{eq9}. $U_1$ can be obtained
  analogously.
\end{proof}

\begin{figure*}[!t] 
\centering
 \renewcommand{\arraystretch}{0.5}

\begin{equation*}
  \mathsmaller{\ket{{\psi_1^0}}} = \begin{array}{c@{\!\!\!}l}
  \left[ \begin{array}[c]{ccccc} 
     \mathsmaller{0}\\
     \mathsmaller{\ast}\\
  \hline
    \mathsmaller{\ast}\\
      \mathsmaller{\ast}\\
    \hline
     \mathsmaller{\ast}\\
     \mathsmaller{\ast}\\
    \hline
    \mathsmaller{\ast}\\
      \mathsmaller{\ast}\\
    \hline
     \mathsmaller{\ast}\\
     \mathsmaller{\ast}\\
    \hline
    \mathsmaller{\ast} \\
      \mathsmaller{\ast}\\
    \hline
     \mathsmaller{\ast}\\
     \mathsmaller{\ast}\\
    \hline
    \mathsmaller{\ast}\\
      \mathsmaller{\ast}\\
  \end{array}  \right]
\end{array}
\phantom{a}
 \xrightarrow{ C_{3}^{\textnormal{u}}(U_{1,0}^{\textnormal{u}}) }  
 \begin{array}{c@{\!\!\!}l}
   \left[ \begin{array}[c]{ccccc} 
     \mathsmaller{0}\\
    \mathsmaller{\ast}\\
     \mathsmaller{0}\\
     \mathsmaller{\ast}\\
    \hline
    \mathsmaller{0}\\
    \mathsmaller{\ast}\\
     \mathsmaller{0}\\
       \mathsmaller{\ast}\\
    \hline
    \mathsmaller{0}\\
    \mathsmaller{\ast}\\
     \mathsmaller{0}\\
          \mathsmaller{\ast}\\
    \hline
    \mathsmaller{0}\\
     \mathsmaller{\ast}\\
    \mathsmaller{0}\\
    \mathsmaller{\ast}\\
\end{array}  \right]
\end{array}
\phantom{a}
 \xrightarrow{C^{\textnormal{u}}_{2}(U_{1,1}^{\textnormal{u}}) }  
 \begin{array}{c@{\!\!\!}l}
   \left[ \begin{array}[c]{ccccc} 
    \mathsmaller{0}\\
    \mathsmaller{\ast}\\
     \mathsmaller{0}\\
      \mathsmaller{\ast}\\
      \hline
     \mathsmaller{0}\\
   \mathsmaller{\ast}\\
  \mathsmaller{0}\\
     \mathsmaller{0}\\
    \hline
     \mathsmaller{0}\\
    \mathsmaller{\ast}\\
   \mathsmaller{0}\\
     \mathsmaller{0}\\
      \hline
     \mathsmaller{0}\\
    \mathsmaller{\ast}\\
   \mathsmaller{0}\\
     \mathsmaller{0}\\
\end{array}  \right]
\end{array}
\phantom{a}
 \xrightarrow{C_{3}(U_{1,1})}  
 \begin{array}{c@{\!\!\!}l}
   \left[ \begin{array}[c]{ccccc} 
      \mathsmaller{0}\\
    \mathsmaller{\ast}\\
     \mathsmaller{0}\\
      \mathsmaller{0}\\
     \mathsmaller{0}\\
   \mathsmaller{\ast}\\
  \mathsmaller{0}\\
     \mathsmaller{0}\\
    \hline
     \mathsmaller{0}\\
    \mathsmaller{\ast}\\
   \mathsmaller{0}\\
     \mathsmaller{0}\\
     \mathsmaller{0}\\
    \mathsmaller{\ast}\\
   \mathsmaller{0}\\
     \mathsmaller{0}\\
\end{array}  \right]
\end{array}
\phantom{a}
 \xrightarrow{ C^{\textnormal{u}}_{1}(U_{1,2}^{\textnormal{u}}) }  
 \begin{array}{c@{\!\!\!}l}
   \left[ \begin{array}[c]{ccccc} 
     \mathsmaller{0}\\
    \mathsmaller{\ast}\\
   \mathsmaller{0}\\
     \mathsmaller{0}\\
    \mathsmaller{0}\\
 \mathsmaller{\ast}\\
     \mathsmaller{0}\\
    \mathsmaller{0}\\
    \hline
         \mathsmaller{0}\\
    \mathsmaller{\ast}\\
   \mathsmaller{0}\\
     \mathsmaller{0}\\
    \mathsmaller{0}\\
 \mathsmaller{0}\\
     \mathsmaller{0}\\
    \mathsmaller{0}\\
\end{array}  \right]
\end{array}
\phantom{a}
 \xrightarrow{C_{3}(U_{1,2}) }  
 \begin{array}{c@{\!\!\!}l}
   \left[ \begin{array}[c]{ccccc} 
      \mathsmaller{0}\\
    \mathsmaller{\ast}\\
   \mathsmaller{0}\\
     \mathsmaller{0}\\
    \mathsmaller{0}\\
 \mathsmaller{0}\\
     \mathsmaller{0}\\
    \mathsmaller{0}\\
         \mathsmaller{0}\\
    \mathsmaller{\ast}\\
   \mathsmaller{0}\\
     \mathsmaller{0}\\
    \mathsmaller{0}\\
 \mathsmaller{0}\\
     \mathsmaller{0}\\
    \mathsmaller{0}\\

\end{array}  \right]
\end{array}
\phantom{a}
 \xrightarrow{C_{3}(U_{1,3}) }  
\begin{array}{c@{\!\!\!}l}
  \left[ \begin{array}[c]{ccccc} 
      \mathsmaller{0}\\
   \mathsmaller{1}\\
     \mathsmaller{0}\\
    \mathsmaller{0}\\
 \mathsmaller{0}\\
     \mathsmaller{0}\\
    \mathsmaller{0}\\
       \mathsmaller{0}\\
    \mathsmaller{0}\\
   \mathsmaller{0}\\
     \mathsmaller{0}\\
    \mathsmaller{0}\\
 \mathsmaller{0}\\
     \mathsmaller{0}\\
    \mathsmaller{0}\\
        \mathsmaller{0}\\
         \end{array}  \right]
  \end{array}
\end{equation*}

\[
\Qcircuit @C=4.75em @R=.05em {
&& \gate{} \qwx[1] \qw 			&	 \gate{} \qwx[1] \qw 					&\ctrlo{1} \qw					&\gate{} \qwx[1] \qw			&\ctrlo{1} \qw					 	&\gate{U_{1,3}} \qw &\qw			\\
&& \gate{} \qwx[1] \qw			&	 \gate{} \qwx[1] \qw  					&\ctrlo{1} \qw					&\gate{U^{\textnormal{u}}_{1,2}}	&\gate{U_{1,2}} \qw 						&\ctrlo{-1} \qw&\qw 		\\
&& \gate{} \qwx[1]  \qw 		&	\gate{U^{\textnormal{u}}_{1,1}} \qw&\gate{U_{1,1}} \qw	&\qw 					&\ctrlo{-1} \qw										&\ctrlo{-1} \qw&\qw					\\
&& \gate{U^{\textnormal{u}}_{1,0}} \qw &\qw						&\ctrl{-1} 						&\qw						&\ctrl{-1} \qw									&\ctrl{-1} \qw&\qw 						 \\
}
\]
\caption{Implementing the second column of an isometry $V$ from
  $m\geqslant 1$ qubits to $n=4$ qubits. The operation of $G_1$ on
  $\ket{{\psi_1^0}}:=G_0V\ket{1}_m$ can be decomposed into operators
  $\{G_1^{i}\}_{i \in {0,1,2,3}}$, where
  $G_1^0=C_{3}^{\textnormal{u}}(U_{1,0}^{\textnormal{u}})$,
  $G_1^{1}=C_{3}(U_{1,1})
  C_{2}^{\textnormal{u}}(U_{1,1}^{\textnormal{u}})$,
  $G_1^{2}=C_{3}(U_{1,2})
  C_{1}^{\textnormal{u}}(U_{1,2}^{\textnormal{u}})$ and
  $G_1^{3}=C_{3}(U_{1,3})$. Note that all these gates act trivially on
  $\ket{0}_n$. The symbol ``$\ast$" denotes an arbitrary complex
  number.}
\label{fig:coloumn2}
\end{figure*}

\renewcommand{\arraystretch}{1}

As noted above, the unitary operation $G_0$ can be decomposed using the reverse of the
decomposition scheme for state preparation as described
in~\cite{10}. First we act with a UCG
$G_0^{0}=C_{n-1}^{\textnormal{u}}(U_{0,0}^{\textnormal{u}})$ on the
least significant qubit. The gate $G_0^{0}$ has a $2 \times 2$ block
diagonal structure. Using Lemma~\ref{rotate_to_basis_state} we can
construct $G_0^{0}$ such that it zeroes every second entry of
$\ket{{\psi_0^0}}:=V\ket{0}_m$ (see Fig.~\ref{fig:coloumn1}). This
corresponds to disentangling (i.e., rotating to product form) the
least significant qubit, so we can write
$G_0^{0}\ket{{\psi_0^0}}=\ket{{\psi_0^1}}\otimes\ket{{0}}$ for some
state $\ket{{\psi_0^1}}\in\cH_{n-1}$. Now we apply the same procedure
to $\ket{{\psi_0^1}}$ leaving the least significant qubit
invariant. We act with
$G_0^{1}:=C_{n-2}^{\textnormal{u}}(U_{0,1}^{\textnormal{u}}) $, which
corresponds to conditionally rotating the second least significant
qubit, leading to
$G_0^{1}G_0^{0}\ket{{\psi_0^0}}=\ket{{\psi_0^2}}\otimes\ket{{0}}\otimes\ket{{0}}$,
for some $\ket{{\psi_0^2}}\in\cH_{n-2}$.  We continue in this fashion
until all the qubits have been disentangled.  Thus we have constructed a
quantum gate $G_0:=G_0^{n-1} G_0^{n-2}\dots G_0^{0}$ such that
$G_0\ket{{\psi_0^0}}=\ket{0}_n$\footnote{Note that $G_0^\dagger$ is a
  circuit for preparing the state $\ket{\psi^0_0}$; in this sense we
  have performed the inverse of state preparation.}.

In the following we describe how to construct a unitary $G_1$ setting
the second column of $G_0V$ to $(0,1,0,\ldots,0)$ without affecting
the first column.  We construct
$G_1^0=C_{n-1}^{\textnormal{u}}(U_{1,0}^{\textnormal{u}})$ choosing
the unitary operations such that the first entry of each pair becomes
zero (see Fig.~\ref{fig:coloumn2}).  In other words, defining
$\ket{{\psi_1^0}}:=G_0V\ket{1}_m$ we have
$G_1^{0}\ket{{\psi_1^0}}=\ket{{\psi_1^1}}\otimes\ket{1}$, for some
state $\ket{{\psi_1^1}}$. Note that, by construction, the first column of $G_0V$ in matrix form is
$(1,0,\ldots,0)$, and, since $G_0$ is unitary, the first row also has
the form $(1,0,\ldots,0)$.  Hence the first entry of
$\ket{{\psi_1^0}}$ is already 0 and we can set
the upper most $2\times 2$ block of the uniformly controlled gate
$G_1^0$, i.e. the block acting on the states $\ket{0}_n$ and
$\ket{1}_n$, to the identity.  Therefore we can perform this step
without affecting the first column, i.e. $G_1^0 G_0V\ket{0}_m= G_1^0
\ket{0}_n=\ket{0}_n$.  The next step would be to do the same to
$\ket{{\psi_1^1}}$ (i.e., zero every second entry). Doing so using a
$C^{\textnormal{u}}_{n-2}(U)$ gate would, in general, have a
non-trivial effect on the basis state $\ket{0}_n$. Therefore we modify
the procedure and instead use a $C^{\textnormal{u}}_{n-2}(U)$ gate to
zero every second entry except that in the upper most double block of
$\ket{{\psi_1^1}}$ or equivalently that in the upper most block of
four elements of $G_1^0 \ket{{\psi_1^0}}$. We subsequently correct for
this using an additional MCG acting on the second least significant
qubit, i.e., we set $G_1^{1}=C_{n-1}(U_{1,1})
C_{n-2}^{\textnormal{u}}(U_{1,1}^{\textnormal{u}})$. With this
additional MCG we can directly address the quantum states
corresponding to the two non zero entries in the upper-most
four-element block. Indeed, controlling on $\ket{0} \otimes \ket{0}
\otimes \dots \otimes \ket{0}$ on the first $(n-2)$ qubits and on
$\ket{1}$ on the least significant qubit we can zero the second non
zero entry of the upper-most four-element block without affecting
$\ket{0}_n$.

We conclude that
$G_1^{1}G_1^{0}\ket{{\psi_1^0}}=\ket{{\psi_1^2}}\otimes\ket{{0}}\otimes\ket{{1}}$
and $G_1^{1}\ket{0}_n=\ket{0}_n$. We continue in this way, until the
most significant qubit is disentangled.  We have therefore constructed
a operation $G_1$ such that
$G_1G_0V\ket{1}_m=G_1\ket{{\psi_1^0}}=\ket{1}_n$ and
$G_1G_0V\ket{0}_m=G_1\ket{0}_n=\ket{0}_n$.

This procedure can be continued in a similar fashion, leading to
unitaries $G_k$ such that $G_k G_{k-1}\dots G_0 V \ket{k}_m=\ket{k}_n$
and $G_k \ket{i}=\ket{i}$ for all $i \in \{0,1,\dots,k-1 \}$. For a
general description of the construction of the unitary $G_k$ see
Appendix~\ref{sec:app_isounif}.  We can hence construct a unitary
operator $G:=G_{2^{m}-1}G_{2^{m}-2}\dots G_{0}$
satisfying $GV=I_{2^n\times2^m}$.

In order to compute the number of \cnots{} used for such a
decomposition, we use the following existing results:
\begin{enumerate}[label=(\roman*)]
\item \label{pt:1} $N_{ \Delta C^{ \textnormal{u} }_{k}}=2^k-1$ \cnots{} are sufficient to decompose a UCG
  with $k$ controls, up to a diagonal gate~\cite{10}.
\item \label{pt:2} $N_{\Delta}(m)=2^m-2$ \cnots{} are sufficient to decompose a
  diagonal gate acting non trivially on $m$ qubits~\cite{diag}.
\item \label{pt:3} $N_{ C_{n-1}(W)}=\mathcal O\left(n\right)$ \cnots{} are sufficient
  to decompose an $(n-1)$-controlled special unitary
  gate $W$~\cite[Corollary 7.10]{5}.
\end{enumerate}

To take advantage of~\ref{pt:1}, we require a small modification to
our decomposition scheme.  Note that instead of implementing the UCGs,
we do so up to diagonal gates, i.e., for every $k$, instead of
$C_{k}^{\textnormal{u}}(U)$ we implement
$\Delta_{k+1}C_{k}^{\textnormal{u}}(U)$, for some diagonal gate
$\Delta_{k+1}$ on $k+1$ qubits. The effect of these diagonal gates is
then be corrected for at the end of the entire circuit by adding a
diagonal gate that acts non-trivially on $m$ qubits and whose \cnot{}
count is given in~\ref{pt:2}.  (In fact, the number of \cnots{}
required for this is of sufficiently low-order that it doesn't feature
in the count of Theorem~\ref{thm_main}.)

Furthermore, as shown in Lemma~\ref{rotate_to_basis_state}, we only
require MCGs $C_{n-1}(W)$ for $W \in SU(2)$, and hence can
use~\ref{pt:3}.  In fact, we have modified the decomposition described
in~\cite{5} and used some technical tricks (see
Appendix~\ref{sec:app_multicontroll}) to obtain a \cnot{} count for a
$C_{n-1}(W)$ gate with leading order $28n$.

We conclude that we can decompose each column of an isometry using at
most   $N_{\textrm{col}}=\sum_{s=0}^{n-1} \left(N_{ \Delta C^{ \textnormal{u} }_{n-1-s}} +N_{ C_{n-1}(W)} \right) =\sum_{s=0}^{n-1} \left(\left(2^{n-1-s}-1\right)+\mathcal
  O\left(n\right)\right)=2^n+\mathcal O\left(n^2\right)$ \cnots. Note
that (for simplicity) we have overcounted the number of additional
MCGs, since in the above we have assumed each $G_k^s$ requires an
additional MCG. Therefore, to decompose an $m$ to $n$ isometry, we
require at most $2^m N_{\textrm{col}}+N_{\Delta}(m)= 2^m\left( 2^n+\mathcal
  O\left(n^2\right)\right)+2^m=2^{m+n}+ \mathcal O\left(n^2
\right)2^m$ \cnots.

Note that we implement every column of the isometry in a similar
fashion.  However, there are a lot of constraints on the last few
columns due to orthogonality, or, in other words, the first $k$
entries of $\ket{{\psi_{k}^0}}:=G_{k-1}G_{k-2} \dots G_0V\ket{k}_m$
are already zero by construction and so we have only to act on the
other $2^n-k$ entries. Therefore one might expect that the \cnot{}
count for $G_k$ decreases when $k$ increases. Since we use $2^n$
\cnots{} to leading order for each column, our decomposition scheme
doesn't take an advantage of this fact (for large $n$). Hence the
column-by-column decomposition has some inefficiency in the case where
$m\simeq n$ (by comparison to the case $m\ll n$). To give an improved
count in the cases $m= n-1$ and $m=n$, we introduce a further
decomposition scheme based on the CSD, which is adjusted to the
unitary structure, in Section~\ref{sec:isocsd}. Note that this scheme
corresponds exactly to the decomposition scheme of~\cite{2} in the
case $m=n$.\smallskip

\begin{rmk} \label{rmk:nearest_neighbour} In some physical
  realizations it is difficult to implement \cnot{} gates between
  non-adjacent qubits. The decomposition in this section can be
  adapted to the gate library containing only nearest neighbour
  \cnot{} and single-qubit gates in a relatively efficient way.  To do
  so, note that the UCGs used to implement one column of an $m$ to $n$
  isometry can be performed with at most $(5/3)2^n +\mathcal
  O\left(n^2\right)$ nearest neighbour \cnot{} gates~\cite{10}.
  Furthermore, since a \cnot{} gate acting between qubits a distance
  $n$ apart can be decomposed using $\mathcal O\left(n\right)$ nearest
  neighbour \cnot{} gates~\cite{2}, the MCGs used to implement one
  column use $\mathcal O\left(n^3\right)$ nearest neighbour \cnot{}
  gates. Therefore the decomposition of an $m$ to $n$ isometry uses at
  most $(5/3)2^{m+n}+ \mathcal O\left(n^3\right) 2^m$ nearest
  neighbour \cnot{} gates.
\end{rmk}

\subsection{Decomposition of isometries using the Cosine-Sine Decomposition} \label{sec:isocsd}

The most efficient known decomposition scheme for arbitrary unitary
operators in term of the number of \cnot{} gates required uses the
CSD~\cite{2}. In this section we adapt the decomposition scheme used
in~\cite{2} to $m$ to $n$ isometries.  To simplify the exposition, the
count given here is not the lowest we can obtain; an improvement is
given in Appendix~\ref{opt_CNOT_count_CSD}.

\begin{thm}
  Let $m$ and $n$ be natural numbers with $2\leqslant m\leqslant n$
  and $V$ be an isometry from $m$ qubits to $n$ qubits. There exists a
  decomposition of $V$ in terms of single-qubit gates and \cnots{} 
  such that the number of \cnot{} gates required satisfies
\begin{equation}\label{eq:CSD_count_1}
  N_{\mathrm{iso}}(m,n)\leqslant 3\cdot2^{2n-3}-2^n+2^{m-4}\left(3\cdot2^m-8\right).
\end{equation}
\end{thm}

The Cosine-Sine Decomposition (CSD)~\cite{csd} was first used
by~\cite{csd_quantum} in the context of quantum computation.  In
particular, the CSD states that every unitary matrix $U\in
\mathbb{C}^{2^n\times 2^n}$ can be decomposed in terms of unitaries
$A_0,A_1,B_0,B_1 \in \mathbb{C}^{2^{n-1}\times 2^{n-1}}$ and real
diagonal matrices $C$ and $S$ satisfying $C^2+S^2=I$:

\begin{equation}\label{eq:CSD_matrix}
U   =  \left(\begin{array}{c|c}A_0 &0 \\\hline 0& A_1\end{array}\right)\left(\begin{array}{c|c}C & -S  \\\hline S & C \end{array}\right)\left(\begin{array}{c|c}B_0 &0 \\\hline 0 & B_1\end{array}\right)
\end{equation}

The CSD can be summarized by the gate identity
\[
\Qcircuit @C=1em @R=.2em {
&&			& \multigate{2}{U_n}	&\qw& &	&& \gate{}\qwx[2]       & \gate{R_{y}} 	&  \gate{}  \qwx[2]	&   \qw	\\
&& 			&	&	&=			&&&       	      &	 			   & &   & 	  \\
&\lstick{n-1}& {\backslash}	& \ghost{U_n}&\qw &&& {\backslash}	&\gate{U_{n-1}}   &\gate{} \qwx[-2] \qw&		\gate{U_{n-1}} 	&   \qw	 \\
}
\]
Together with
\begin{equation}\label{CSD2}
\Qcircuit @C=0.65em @R=.2em {
&&&			& \gate{}\qwx[2]	&\qw& &&&	&& \qw       & \gate{R_z} 	&   \qw&   \qw	\\
&&& 			&		&&&=			&&&&       	      &	 			 \\
&&\lstick{n-1}& {\backslash}	& \gate{U_{n-1}}&\qw &&&&& {\backslash}	&\gate{U_{n-1}}   &\gate{} \qwx[-2] \qw&  \gate{U_{n-1}} &  \qw	 \\
}
\end{equation}
(which is Theorem~12 of~\cite{2}) it allows a recursive decomposition
of an arbitrary unitary operation in terms of single-qubit gates and
uniformly controlled $R_y$ and $R_z$ gates.

In the case of an isometry, we again use a representation in terms of
a unitary matrix, $V_n$, such that $V=V_n I_{2^n\times2^m}$.  Now, if
$n>m$, we can take the control qubit of the first
$(n-1)$-qubit-$C^{\textnormal{u}}_{1}(U_{n-1})$ gate to be in the
state $\ket{0}$, and hence this gate need not be uniformly
controlled. Thus, the following circuit identity holds
\[
\Qcircuit @C=0.8em @R=.2em {
&\lstick{\ket{0}}&			& \multigate{2}{V_n}	&\qw& &&	&\lstick{\ket{0}}& \qw       & \gate{R_{y}} 	&  \gate{}  \qwx[2]	&   \qw	\\
&& 			&		&&=			&&&&       	      &	 			   & &   & 	  \\
&\lstick{n-1}& {\backslash}	& \ghost{ V_n}&\qw &&&& {\backslash}	&\gate{V_{n-1}}   &\gate{} \qwx[-2] \qw&		\gate{U_{n-1}} 	&   \qw	 \\
}
\]

Note that $V_{n-1}$ represents an $m$ to $n-1$ isometry. In the matrix
representation the circuit identity above corresponds to setting
$B_1=B_0$ in equation~(\ref{eq:CSD_matrix}). We can decompose the
$(n-1)$-qubit-$C^{\textnormal{u}}_{1}(U)$ gate as above so that
\[
\Qcircuit @C=0.33em @R=.2em {
&&\lstick{\ket{0}}&			& \multigate{2}{ V_n}	&\qw& &&&&&&& 	&\lstick{\ket{0}}& \qw       & \gate{R_{y}} 	&   \qw&  \gate{R_z} 	&  \qw  & \qw	\\
&&& 			&		&&&=&&&			& &&&&       	      &	 			   & &  			 & &&	  \\
&&\lstick{n-1}& {\backslash}	& \ghost{ V_n}&\qw &&&&&&&& & {\backslash}	&\gate{V_{n-1}}   &\gate{} \qwx[-2] \qw&  \gate{U_{n-1}} & \gate{}  \qwx[-2]&   \gate{U_{n-1}}	&   \qw	 \\
}
\]
 
We can use this idea to recursively decompose $V_n$. The uniformly
$(n-1)$-controlled rotations can be decomposed using at most $2^{n-1}$
\cnot{} gates~\cite{diag,unif_rot}. The two $U_{n-1}$ gates can be
decomposed by using the CSD and the circuit equivalence~\eqref{CSD2} recursively until
two-qubit gates remain\footnote{We could finish the recursion at any stage, such that only  $\tilde{n}$-qubit unitaries reamain. Therefore, an improvement of the \cnot{} count for $\tilde{n}$-qubit unitaries could help to improve the  \cnot{} count given in equation~(\ref{eq:CSD_count_1}) (and equation~(\ref{eq:CSD_count_2})).} (each of which can be implemented with 3
\cnots). In this way it can be shown that each $U_{n-1}$ requires at
most $(9/16)4^{n-1}-(3/2)2^{n-1}$ \cnot{} gates~\cite{2}. Note that
this is not the optimal count reached in~\cite{2}, but we use this
slightly weaker count here for simplicity (a count that takes into
account the additional optimizations of the Appendix of~\cite{2} can
be found in Appendix~\ref{sec:app_isocsd}). The \cnot{} count for an
$m$ to $n$ isometry, $N_{\mathrm{iso}}(m,n)$, hence satisfies the
recursion relations
\begin{align}
\label{eq:rec1} 
		N_{\mathrm{iso}}(m,i+1)&=N_{\mathrm{iso}}(m,i)+\frac{9}{8}4^{i}-2^{i},
                \text{ if } m\leqslant i < n\, ,\\
N_{\mathrm{iso}}(m,m)&=\frac{9}{16}4^m-\frac{3}{2}2^m\, .\label{eq:recs1}
\end{align}
Solving these leads to the claimed count.\smallskip

\begin{rmk} [CSD approach zeroes too many entries] 
  Recall that constructing a gate $V_n$ such that $V=V_n
  I_{2^n\times2^m}$ is equivalent to constructing a gate $V_n^{\dagger}$
  such that $V_n^{\dagger}V= I_{2^n\times2^m}$. Therefore,
  rewriting equation~(\ref{eq:CSD_matrix}), the first recursion step of the
 CSD approach leads to

\begin{equation}\label{eq:CSD_matrix_inv}
\left(\begin{array}{c|c}C & S  \\\hline -S & C \end{array}\right) \left(\begin{array}{c|c}A_0^{\dagger} & 0 \\\hline 0& A_1^{\dagger} \end{array}\right)  U   = \left(\begin{array}{c|c}B_0 &0  \\\hline 0& B_1\end{array}\right)
\end{equation}

If $m<n-1$ we apply the same procedure to $B_0$.  However,
in this case, we already zeroed more entries than necessary in
the first recursion step. Specifically, it was unnecessary to zero at
least half of the entries in the upper right  and in the lower left 
$2^{n-1}\times 2^{n-1}$-dimensional block  of the matrix on the rhs
of equation~\eqref{eq:CSD_matrix_inv}, and the number of unnecessary zeros
grows as $m$ decreases. This intuitively explains why the CSD approach
is not well-suited to $m$ to $n$ isometries, where $m<n-1$: by zeroing
too many entries, more \cnot{} gates are used than needed.

\end{rmk}

\begin{rmk} [Optimized state preparation] \label{spodd} As a
  by-product of the above we obtain an improved bound over that
  of~\cite{10} on the number of \cnot{} gates required for state
  preparation on an odd number $n=2k+1\geqslant5$ of qubits.  The
  optimized decomposition is based on~\cite{3} and described in
  Section~\ref{opt_SP_app}.  The count~(\ref{CNOT_SP_odd}) using state
  preparation on $k$ qubits, which requires $2^k-k-1$ \cnots{} (as
  in~\cite{10}), gives the following count for state preparation
  starting from the basis state $\ket{0}^{\otimes n}$:
\begin{equation}
		N_{\mathrm{SP_{opt}}}(n)\leqslant\frac{23}{24} 2^n-\frac{3}{2}2^{\frac{n+1}{2}}+4/3.
	\label{eq:statepreparation_opt}
\end{equation}

Previously, the bound of $\frac{23}{24} 2^n$ \cnots{} to leading order was only known to be achievable for an even number of qubits~\cite{3} with a slightly weaker bound of $2^n$ \cnots{} to leading order in the odd case~\cite{10}.  For completeness, the bound for even $n$ is~\cite{3}
$$N_{\mathrm{SP}}(n)\leqslant\frac{23}{24}2^n-2^{\frac{n}{2}+1}+5/3$$
and since this bound is larger than~\eqref{eq:statepreparation_opt} for all $n$, the bound in the even case can be used for all $n$.

It is interesting to note the parallelizability of our circuit for state preparation, similarly to~\cite{3}. The form of the circuit means that, for large (odd) $n$, the circuit depth (i.e., the number of computational steps needed to perform the circuit) is about $3/4$ of the total gate count. ÊMeasuring the circuit depth only in terms of \cnots{}, our decomposition scheme has depth $\frac{23}{32}2^n$ to leading order, improving the previous best known bound of $\frac{23}{24}2^n$~\cite{3}. In the case of even $n$, the minimum known circuit depth is $\frac{23}{48}2^n$~\cite{3}.
\end{rmk}

\section{Comparison of decompositions}\label{sec:results}

We introduced three constructive decomposition schemes for arbitrary
isometries from $m$ to $n$ qubits and derived a lower bound on the
number of \cnot{} gates required for such decompositions. The
asymptotic results are summarized in Tables~\ref{tab:Results}
and~\ref{tab:Results2}. To compare the three decomposition schemes, we
consider the ratios $c_K(m,n)$, $c_{CC}(m,n)$ and $c_{CSD}(m,n)$ of
the \cnot{} count for the optimized decomposition scheme of Knill, the
column-by-column approach or the CSD approach, respectively, to that
of the lower bound for an $m$ to $n$ isometry. First note that for
$m\geqslant 5$ and for large enough $n$ the optimized decomposition
scheme of Knill performs similarly to the column-by-column
decomposition (i.e., $c_K(m,n) \simeq c_{CC}(m,n)$). For $m\leqslant
4$ we have $c_{CC}(m,n)Ê\simeq 2$ and $c_K(m,n)$ varies between
$c_K(4,n)\simeq 2$ (if $n$ is even) and $c_K(0,n)\simeq 4.8$ (if $n$
is odd). Hence the column-by-column decomposition requires fewer
\cnot{} gates if $m \leqslant 4$ (and $n$ is large). In the case
$m\simeq n$, the CSD approach may outperform the other two
decompositions. For any natural number $d$ and for sufficiently large
$n$, we have $c_{CC}(n-d,n)=2^{d+2}/(2^{d+1}-1)$ (and
$c_{CC}(n-d,n)\simeq c_K(n-d,n)$) and
$c_{CSD}(n-d,d)=\frac{23(2^{2d+1}+1)}{36(2^{d+1}-1)}$. In particular
$c_{CC}(n-2,n) \simeq 2.3$ and $c_{CC}(n-1,n) \simeq 2.7$ for large
$n$. For $m=n-1$ we can use the CSD approach to again reach
$c_{CSD}(n-1,n)\simeq 1.9$ for large $n$.

The column-by-column decomposition and the CSD-approach also perform
well for small $m$ and $n$. We give a step by step description of how to
decompose $m$ to $n\leqslant 4$ isometries in
Appendix~\ref{small_cases_appendix}. The results are summarized in
Table~\ref{Tab:small_cases}.

In addition we could use the CSD-approach (and a technical trick) to lower the \cnot{} count for state preparation. In particular we could lower the lowest
known \cnot{} count for state preparation on $4$ qubits from
$9$~\cite{3} to $8$ \cnots{} and on $5$ qubits from $26$~\cite{3,10}
to $19$ \cnots{} (cf.\ Appendix§~\ref{opt_SP_app}).

The column-by-column decomposition performs similarly to the optimized
decomposition of Knill with respect to the \cnot{} count, but there
are other differences that should be noted. For example, the
column-by-column decomposition adapts quite well to implementations
where we only allow nearest neighbour \cnot{} gates (cf.\
Remark~\ref{rmk:nearest_neighbour}). The optimized decomposition
scheme of Knill has the advantage that some of the gates can be
performed in parallel (cf.\ the circuit diagrams in
Section~\ref{sec:Knill}).

Another important difference between the column-by-column
decomposition and the optimized decomposition of Knill is their
dependence on the efficiency of the decomposition of their building
blocks. In the first case, any improvement of the leading order of the
\cnot{} count of uniformly controlled gates (up to diagonal gates)
leads to an improvement of the leading order of the \cnot{} count for
isometries (cf.\ Theorem~\ref{thm_main}). Where in the second case,
the leading order of the \cnot{} count depends on the leading order of
the \cnot{} count for arbitrary unitary gates (cf.\
Theorem~\ref{thm_Knill}).

\begin{rmk}
  Another interesting black box relation can be extracted
  from~\cite{Wolf}, where the Sinkhorn normal form for unitary
  matrices is used to decompose a unitary into a sequence of diagonal
  gates and discrete Fourier transforms (cf.\ Corollary 1
  of~\cite{Wolf}). Since we can perform the discrete Fourier transform
  with a polynomial number of gates, they do not contribute to the leading
  order of the \cnot{} count of this decomposition. Therefore, this
  decomposition allows us to relate the efficiency with which we can
  decompose a unitary with the decomposition of diagonal gates.
\end{rmk}

\section{Application to quantum operations and POVMs} \label{sec:ACO}
Experimental groups strive to demonstrate their ability to control a
small number of qubits, and the ultimate demonstration would be the
ability to do any quantum operation on them (i.e., any completely
positive trace-preserving (CPTP) map).  Since any such operation can
be implemented via an isometry followed by partial trace (using
Stinespring's theorem), we can use our decomposition scheme for
isometries to efficiently synthesize arbitrary CPTP maps.

Indeed, we can use a similar parameter counting argument as used to
derive the lower bound for isometries to find a lower bound on the
number of \cnot{} gates required to implement arbitrary CPTP maps via
a fixed quantum circuit topology.  First we use the Choi-Jamiolkowski
isomorphism~\cite{Pillis, Jamiolkowski, Choi}  to simplify the parameter count.
This isomorphism states that the set of all CPTP maps from a system
$A$ consisting of $m$ qubits to a system $B$ consisting of $n$ qubits
is isomorphic to the set of all density operators $\rho_{AB}$ on
$\mathcal{H}_A\otimes \mathcal{H}_B$ satisfying
$\tr_B(\rho_{AB})=\frac{1}{2^m}I_A$. Since a density operator
$\rho_{AB}$ is Hermitian, it can be described by $2^{2(n+m)}$ real
parameters. The condition $\tr_B(\rho_{AB})=\frac{1}{2^m}I_A$
corresponds to $2^{2m}$ constraints, and hence the determination of a
CPTP map requires $2^{2(n+m)}-2^{2m}$ real parameters.

We restrict our analysis of the lower bound to the following setting:
For the implementation of a CPTP map $\mathcal{E}$ from an $m$-qubit
system $A$ to an $n$-qubit system $B$ we allow the use of an arbitrary
number $k$ of qubits on which we can perform \cnot{} and single-qubit
gates, before we trace out a system $C$ consisting of $k-n$
qubits. (Since tracing out qubits commutes with quantum gates on the
other qubits, without loss of generality, we can defer tracing out to
the end of the circuit.) We then use a similar argument as used to
derive the lower bound for isometries, but instead of commuting the
$R_x$ and $R_z$ gates to the left of each \cnot{}, we commute them to
the right so that we perform arbitrary single-qubit unitaries on all
of the qubits at the end of the circuit (reversing the order of
circuit diagram~(\ref{comm_to_left})). Since we have unitary freedom
on the system $C$ (because $\tr_C((I_B\ot U_C)\rho_{BC}(I_B\ot
U_C^\dagger))=\tr_C(\rho_{BC}))$, the single-qubit gates on each qubit
of the system $C$ at the end of the circuit cannot introduce
additional parameters. Hence, using $r$ \cnots{}, we can introduce at
most $4r+3n$ real parameters. By the parameter count for a CPTP map
given above, we conclude that a circuit topology has to consist of at
least $ \lceil \frac{1}{4}4^m(4^n-1)-\frac{3}{4}n \rceil$
\cnots{} in order that it can implement arbitrary CPTP maps from $m$
to $n$ qubits\footnote{For a more rigorous proof one could use a
  similar argument as given in~\cite{unitary_lowerb1,
    unitary_lowerb2}.}.

By Stinespring's theorem, every CPTP map $\mathcal{E}$ from an
$m$-qubit system $A$ to an $n$-qubit system $B$ can be implemented
with an isometry $V$ from system $A$ to system $BC$, where the system
$C$ consists of (at most) $n+m$ qubits, followed by partial trace on
$C$. We can use the column-by-column approach\footnote{The optimized
  decomposition scheme of Knill also leads to a similar asymptotic
  result if $m\geqslant 5$.} to decompose the isometry $V$, which
requires $4^{m+n}-\frac{1}{24}2^{2n+m}$ \cnots{} to leading order
(without exploiting the unitary freedom on $C$). Therefore we have
found a way to implement an arbitrary quantum channel from $m$ to $n$
qubits in a constructive and exact way using about four times the
number of \cnots{} required by the lower bound (for large enough $n$).

Note that the results of this section are derived in the setting where
the CPTP map is implemented in the quantum circuit model. However,
this is not the only possibility. For example, alternative methods for
the implementation of quantum channels are described in~\cite{Wang}
and~\cite{piani_linear-optics_2011}, which allow for additional
classical randomness.  In future work we will investigate how to use
our approach in an alternative model that allows either measurements
or classical randomness as additional resources, in order to further
improve the \cnot{} counts.

Note also that, by Naimark's theorem, any POVM on a system $A$ can be
implemented using an isometry from system $A$ to an enlarged system $AB$
followed by a measurement on system $B$. Therefore our decomposition
schemes for isometries can also be used for the implementation of
arbitrary POVMs.

\section{Acknowledgements}
Part of this work was carried out while MC and RC were with ETH Z\"urich. MC was supported by a Sapere Aude grant of the Danish Council for Independent Research, an ERC Starting Grant, the CHIST-ERA project ``CQC'', an SNSF Professorship, the Swiss NCCR ``QSIT'' and the Swiss SBFI in relation to COST action MP1006. JH was also supported by the Swiss NCCR ``QSIT''. RC acknowledges support from the EPSRC's Quantum Communications Hub (grant no.\ EP/P016588/1).

We thank Vadym Kliuchnikov for kindly pointing out reference~\cite{Knill}.

The authors are grateful to the authors of {\tt quant-ph/0406003}, whose
package Qcircuit.tex was used to produce the circuit
diagrams.

\appendix

\section{Technical details} \label{sec:TD}
In this section we give a rigorous proof that the column-by-column decomposition works for arbitrary $m$ to $n$ isometries and  we give an explicit \cnot{} count in the case $n\geqslant 8$. Since MCG arise in the column-by-column decomposition, we first optimize the decomposition of such gates, based on the decomposition scheme of~\cite{5}.  In addition we perform some optimizations for the CSD-approach (based on the Appendix of~\cite{2}) and for state preparation. 

\subsection{Decomposition of MCGs} \label{sec:app_multicontroll}

In this section we describe how to efficiently decompose MCGs
$C_{n-1,n}(U)$, where we focus on the special case of $C_{n-1,n}(W)$
gates, where $W \in SU(2)$. The decomposition schemes are based on
those in~\cite{5}, except that we use some technical tricks to reduce
the number of \cnots{} needed. Note that the number of \cnots{}
required is the same whether we control on one or zero, because we can
always transform a gate controlled on $\ket{0}$ on a certain
control-qubit of a MCG into a gate controlled on $\ket{1}$ using two
$\sigma_{x}$ gates, as illustrated below.
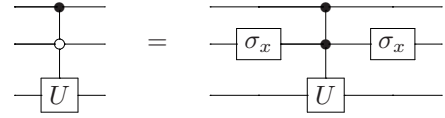
\[
\Qcircuit @C=1.0em @R=.7em {
  &    \ctrl{1} \qw &	\qw 		 & &       &  &     &\qw 		 &    \ctrl{1} \qw &	\qw & \qw  \\
   &  \ctrlo{1}  \qw&	\qw 	 & & =      & &       &  \gate{\sigma_{x}}		 & \ctrl{1}  \qw&   \gate{\sigma_{x}}	&	\qw  \\
  & \gate{U}  & \qw           & &           & &    &\qw 	& \gate{U}  & \qw   & \qw 		& &   \\
}
\]

We denote a $k$-controlled \nt{} gate acting on $n$ qubits by
$C_{k,n}(\sigma_{x})$. In the case $k=2$ with control on
$\ket{1}\ot\ket{1}$, we call such a gate a Toffoli gate.

\begin{lem} [$C_{1,2}(U)$ gates~{\cite[Corollary 5.3]{5}}]  \label{cor5.3}
Any  $C_{1,2}(U)$ gate can be decomposed using two \cnot{} gates, three special unitary gates $A$, $B$ and $C$ and a diagonal gate of the form $E=\ketbra{0}{0}+e^{\I \delta}\ketbra{1}{1}$, where $\delta \in \mathbb{R}$.

\[
\Qcircuit @C=0.8em @R=.47em {
   &    \ctrl{2} \qw & \qw  && 		&& \gate{E}   &  \ctrl{2}  \qw         & \qw  &   \ctrl{2}  \qw        	 &  \qw	& \qw	&\\
   & &         &&=            \\
     &\gate{U}  &	\qw 	&&		        &&   \gate{C}  &  \targ           &\gate{B}      &  \targ    & \gate{A}   &   \qw &\\
}
\]

\end{lem}

\begin{lem} [$C_{2,3}(U)$ gates~{\cite[Lemma 6.1]{5}}]  \label{lem6.1}
Any $C_{2,3}(U)$ gate can be decomposed as follows
\[
\Qcircuit @C=1.0em @R=.7em {
   &    \ctrl{1} \qw & \qw  && 		&&  \qw &  \ctrl{1}  \qw         & \qw  &   \ctrl{1}  \qw        	 & \ctrl{2} \qw	& \qw	&\\
   &  \ctrl{1} \qw &\qw        &&=             && \ctrl{1}& \targ                 &      \ctrl{1}&   \targ                & \qw  	       & \qw		&\\
     &\gate{U}  &	\qw 	&&		        &&  \gate{V}  &  \qw           &\gate{V^{\dag}}      &  \qw     & \gate{V}   &   \qw &\\
}
\]
where $V^2=U$.

\end{lem}

\begin{lem} [Toffoli gates~{\cite[Section~VI A]{5}}] \label{lem1c2} A
  Toffoli gate can be performed with 6 \cnots{} using the following
  circuit
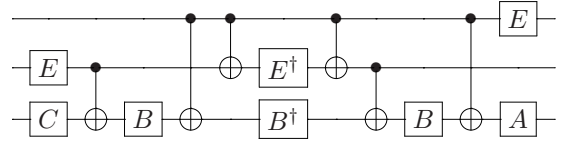
\[
\Qcircuit @C=0.65em @R=.47em {
&\qw     &\qw     &\qw     &\ctrl{2}&\ctrl{1}&\qw           &\ctrl{1}&\qw&\qw&\ctrl{2}&\gate{E}&\qw\\
&\gate{E}&\ctrl{1}&\qw     &\qw     &\targ   &\gate{E^\dagger}&\targ   &\ctrl{1}&\qw&\qw&\qw&\qw\\
&\gate{C}&\targ   &\gate{B}&\targ   &\qw     &\gate{B^\dagger}&\qw&\targ&\gate{B}&\targ&\gate{A}&\qw
}
\]
where $A=R_z(-\frac{\pi}{2})R_y(\frac{\pi}{4})$, $B=R_y(-\frac{\pi}{4})$, $C=R_z(\frac{\pi}{2})$ and $E=\proj{0}+e^{\frac{\I\pi}{4}}\proj{1}$.
\end{lem}

\begin{rmk} [{\cite[Corollary~6.2]{5}}] \label{rmk:6}
By adjusting $A$, $B$, $C$ and $E$, the circuit topology in
Lemma~\ref{lem1c2} can be used to generate $C_{2,3}(U)$ for any
unitary $U$.
\end{rmk}
\begin{proof}
  This circuit equivalence follows from Lemma~\ref{cor5.3} and
  Lemma~\ref{lem6.1} together with the following circuit identities.

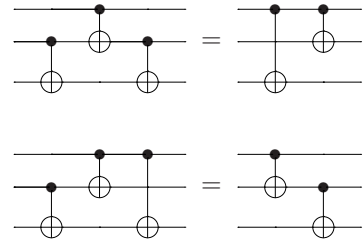
\[
\Qcircuit @C=1.0em @R=.7em {
   &    \qw 		&\ctrl{1} \qw  	& \qw		&\qw	&	&&\ctrl{2}&\ctrl{1}&\qw    \\
   &  \ctrl{1} \qw &\targ 		& \ctrl{1} \qw	&\qw&=	&& \qw& \targ&\qw\\
    &\targ  		&\qw  		&\targ		&\qw	&	&&\targ & \qw&\qw \\
}
\]

\[
\Qcircuit @C=1.0em @R=.7em {
   &    \qw 		&\ctrl{1} \qw  	& \ctrl{2}\qw		&\qw	&	&&\ctrl{1}&\qw&\qw    \\
   &  \ctrl{1} \qw &\targ 		&  \qw	&\qw&=	&&\targ &\ctrl{1} &\qw\\
    &\targ  		&\qw  		&\targ		&\qw	&	&& \qw& \targ&\qw \\
}
\]
\end{proof}

We can halve the \cnot{} count if we are only interested in performing
the Toffoli gate up to a diagonal gate.

\begin{lem} [{\cite[Section~VI B]{5}}] \label{lem1c} Let
  $A:=R_{y}\left(\frac{\pi}{4}\right)$. We can decompose a Toffoli
  gate up to a diagonal gate with the following decomposition

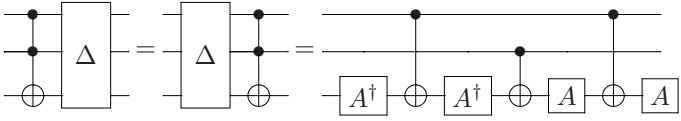
\[
\Qcircuit @C=0.65em @R=.47em {
   &    \ctrl{1} \qw &   \multigate{2}{ \Delta}  &  \qw  & 	&&  \multigate{2}{ \Delta} 	  &    \ctrl{1} \qw		&\qw& &&  \qw &  \ctrl{2}  \qw         & \qw  &  \qw        		&	\qw 	 & \ctrl{2} \qw	& \qw	&\\
   &  \ctrl{1} \qw  & \ghost{ \Delta}	 &\qw             &=	&&\ghost{ \Delta}	&  \ctrl{1} \qw  		&\qw&= && \qw&  \qw               &      \qw        &       \ctrl{1}  & \qw  	       & \qw		& \qw	&\\
     &\targ     & \ghost{ \Delta}	&	\qw 			&	&&\ghost{ \Delta}	&\targ     	  	&\qw& && 	 \gate{A^{\dag}}  &\targ      &\gate{A^{\dag}}      & \targ  &   \gate{A}  & \targ  &   \gate{A}&\\
}
\]

\end{lem}
\begin{proof}  To see this, note that if the second control-qubit is in the state $\ket{0}$, the least significant qubit is unchanged, since $AA^{\dagger}=I$. If the second control-qubit is in the state $\ket{1}$ and the first control-qubit in the state $\ket{0}$, the action on the least significant qubit is $A^2\sigma_x {A^{\dagger}}^2$, which is $-\ketbra{0}{0}+\ketbra{1}{1}$. If both control-qubits are in the state $\ket{1}$, the action on the least significant qubit is $A\sigma_x A \sigma_x A^{\dagger} \sigma_x A^{\dagger}= \sigma_x$. We choose the diagonal gate $ \Delta$  such that  $\ket{010}$ is mapped to $-\ket{010}$.
\end{proof}

\begin{lem} [Diagonal gates commute with UCGs]  
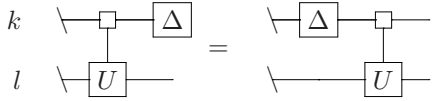
\label{commutation}

\[
\Qcircuit @C=1.0em @R=.47em {
&\lstick{k}&{\backslash} &  \gate{} \qwx[2] \qw&\gate{\Delta}\qw
& &  & {\backslash}   &\gate{\Delta}\qw 		& \gate{}
\qwx[2] \qw	& \qw  \\
&&&&&=&&&& \\
&\lstick{l}& {\backslash} 		      &\gate{U} &\qw 		          	&&     	& {\backslash} 	&\qw		 & \gate{U}  & \qw \\
}
\]
\end{lem}
\begin{proof}  By inspection. \end{proof}

\begin{lem} [$C_{k,n}(\sigma_{x})$,
  $k\leqslant\lceil\frac{n}{2}\rceil$] \label{prop1c} Let $n \geqslant
  5$ denote the total number of qubits considered and $k \in \{
  1,\dots, \lceil \frac{n}{2} \rceil\}$, then we can implement a
  $C_{k,n}(\sigma_{x})$ gate with at most $(8k-6)$ \cnots.
\end{lem}

Note that the case $k=1$ is trivial and the case $k=2$ is implied by Lemma~\ref{lem1c2} (although we know of a tighter bound in both cases).

To illustrate the idea in the remaining cases, consider the
decomposition leading to the desired \cnot{} count for $k=4$, $n=7$.
Lemma~7.2 of~\cite{5} shows that \scalefont{1}
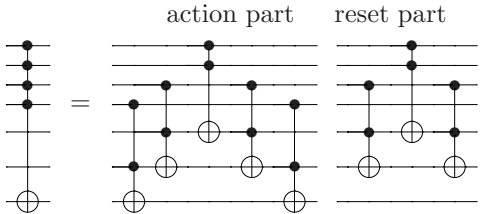
{
\[
\Qcircuit @C=0.4em @R=.5em {
&&			&	&&&&&&		&     		&&&  \mbox{action part}										        	 	& 								 &      		&&       	               &        &        & &           	    &	\mbox{reset part}					&     &\\
&&		&	&&&&&&		&     		&	&									&         	 	& 								 &      		&&       	               &        &            &						& &           	&     &\\
   & \ctrl{1}\qw&\qw& &&&&&&	 \qw			& \qw
   &\qw								&
   \ctrl{1}\qw &	 \qw & \qw 		&\qw
   &\qw 		&\qw&& &\qw
   &		 \qw 	 	& \qw & \ctrl{1} \qw  &\qw									&\qw  	 	&\qw&		\\
   & \ctrl{1} \qw &\qw& &&&&&&	 \qw			& \qw
   &\qw								&
   \ctrl{3}\qw 	& \qw 	&\qw 		&\qw
   &\qw 		&\qw&& & \qw
   &	\qw 		& \qw & \ctrl{3}\qw 	&\qw								&\qw	 	&\qw&		\\
   & \ctrl{1} \qw &\qw&&&&&&& 	 \qw			& \ctrl{2}\qw 		&\qw &\qw 		&\qw	& \ctrl{2}\qw    &\qw&\qw			 &\qw &&&\qw &     \ctrl{2}\qw	 & \qw&\qw 		&\qw & \ctrl{2}\qw     &\qw&	\\
   & \ctrl{3} \qw &\qw&&&=&&&&	 \ctrl{2}\qw	&\qw
   &\qw	      &\qw 		&\qw	&\qw		&\qw
   &\ctrl{2}\qw	&\qw &&&  \qw		  &	\qw 		& \qw &\qw 		&\qw		&\qw		&\qw&					 \\
   &\qw &\qw&& &&&&&		 \qw			&   \ctrl{1}\qw 		&\qw         &\targ			&\qw	& \ctrl{1}\qw 	& \qw		&\qw 		&\qw &&&  \qw	  &       \ctrl{1}\qw& \qw &\targ			&\qw		& \ctrl{1}\qw 	 &\qw&            \\
   &\qw &\qw& &&&&&& 		 \ctrl{1}\qw 	&\targ		 	&\qw		& \qw 	 	&\qw 							&	\targ		&\qw		&\ctrl{1}\qw 	&\qw &&&  \qw	 	 & 	\targ		& \qw 							&\qw 		&\qw		&\targ	        &\qw&									\\
   &\targ &\qw&&&&&&&		 \targ		& 	\qw 		 	&\qw								&\qw 		&\qw 							&	\qw 		&	\qw						   		&\targ		&\qw &&& \qw								 &	\qw 		&\qw 							 &\qw 		&\qw 								&\qw 		 &\qw&								\\			 
   }
\]
}

However, we consider instead the alternative decomposition
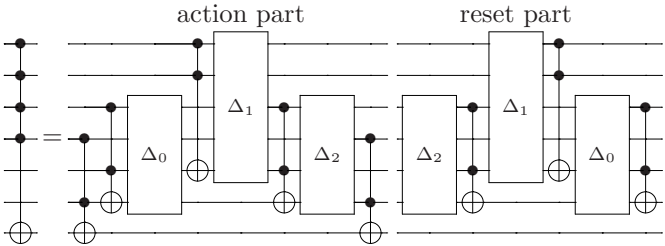
\scalefont{1}{
\[
\Qcircuit @C=0.2em @R=.3em {
&&			&	&&&&&&		&     		&&&  \mbox{action part}										        	 	& 								 &      		&&       	               &        &        & &           	    &	\mbox{reset part}					&     &\\
&&		&	&&&&&&		&     		&	&									&         	 	& 								 &      		&&       	               &        &            &						& &           	&     &\\
   & \ctrl{1}\qw&\qw& &&&&&&	 \qw			& \qw  		 	&\qw								& \ctrl{1}\qw &	 \multigate{4}{\mathsmaller{\Delta_{1}}}& \qw 		&\qw								&\qw 		&\qw&& &\qw							        &		 \qw 	 	& \multigate{4}{\mathsmaller{\Delta_{1}}}& \ctrl{1} \qw  &\qw									&\qw  	 	&\qw&		\\
   & \ctrl{1} \qw &\qw& &&&&&&	 \qw			& \qw 		 	&\qw								& \ctrl{3}\qw 	& \ghost{\mathsmaller{\Delta_{1}}}	&\qw 		&\qw								&\qw 		&\qw&& & \qw							        &	\qw 		& \ghost{\mathsmaller{\Delta_{1}}}& \ctrl{3}\qw 	&\qw								&\qw	 	&\qw&		\\
   & \ctrl{1} \qw &\qw&&&&&&& 	 \qw			& \ctrl{2}\qw 		&\multigate{3}{\mathsmaller{\Delta_{0}}} &\qw 		&\ghost{\mathsmaller{\Delta_{1}}}	& \ctrl{2}\qw    &\multigate{3}{\mathsmaller{\Delta_{2}}}&\qw			 &\qw &&&\multigate{3}{\mathsmaller{\Delta_{2}}} &     \ctrl{2}\qw	 & \ghost{\mathsmaller{\Delta_{1}}}&\qw 		&\multigate{3}{\mathsmaller{\Delta_{0}}} & \ctrl{2}\qw     &\qw&	\\
   & \ctrl{3} \qw &\qw&&&=&&&&	 \ctrl{2}\qw	&\qw 		 	&\ghost{\mathsmaller{\Delta_{0}}}	      &\qw 		&\ghost{\mathsmaller{\Delta_{1}}}	&\qw		& \ghost{\mathsmaller{\Delta_{2}}}		&\ctrl{2}\qw	&\qw &&&  \ghost{\mathsmaller{\Delta_{2}}}		  &	\qw 		& \ghost{\mathsmaller{\Delta_{1}}}&\qw 		&\ghost{\mathsmaller{\Delta_{0}}}		&\qw		&\qw&					 \\
   &\qw &\qw&& &&&&&		 \qw			&   \ctrl{1}\qw 		&\ghost{\mathsmaller{\Delta_{0}}}         &\targ			&\ghost{\mathsmaller{\Delta_{1}}}	& \ctrl{1}\qw 	& \ghost{\mathsmaller{\Delta_{2}}}		&\qw 		&\qw &&&  \ghost{\mathsmaller{\Delta_{2}}}		  &       \ctrl{1}\qw& \ghost{\mathsmaller{\Delta_{1}}} &\targ			&\ghost{\mathsmaller{\Delta_{0}}}		& \ctrl{1}\qw 	 &\qw&            \\
   &\qw &\qw& &&&&&& 		 \ctrl{1}\qw 	&\targ		 	&\ghost{\mathsmaller{\Delta_{0}}} 		& \qw 	 	&\qw 							&	\targ		&\ghost{\mathsmaller{\Delta_{2}}}		&\ctrl{1}\qw 	&\qw &&&  \ghost{\mathsmaller{\Delta_{2}}}	 	 & 	\targ		& \qw 							&\qw 		&\ghost{\mathsmaller{\Delta_{0}}}		&\targ	        &\qw&									\\
   &\targ &\qw&&&&&&&		 \targ		& 	\qw 		 	&\qw								&\qw 		&\qw 							&	\qw 		&	\qw						   		&\targ		&\qw &&& \qw								 &	\qw 		&\qw 							 &\qw 		&\qw 								&\qw 		 &\qw&								\\			 
   }
\]
}

To see that this is also valid, note that the diagonal gates
$\Delta_i$ are of the same kind as introduced in Lemma~\ref{lem1c} and
therefore $\Delta_i=\Delta_i^{\dagger{}}$. By Lemma~\ref{commutation}
the two $\Delta_2$ and $\Delta_1$ gates cancel each other out. In
addition, the combination of all gates between the two $\Delta_0$
gates together correspond to a UCG acting only on the least
significant (lowest) qubit, and hence the two $\Delta_0$ gates cancel
out each other by Lemma~\ref{commutation}.

The Toffoli gates that don't act on the least significant qubit, can
be decomposed together with the diagonal gates using
Lemma~\ref{lem1c}. This leads to the following decomposition of the
  action part of the last circuit
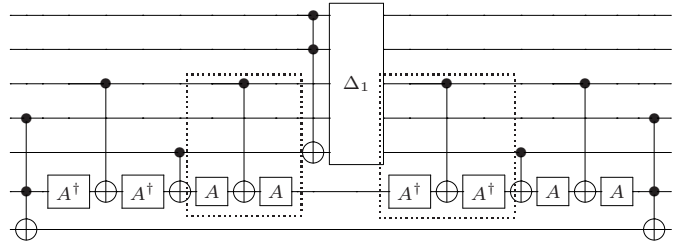
\[
\Qcircuit @C=0.2em @R=.4em {
  &  \qw&  \qw					&\qw	&\qw	&\qw		&\qw			&\qw	&\qw&\qw	&\qw									&\ctrl{1}\qw&\multigate{4}{\mathsmaller{\Delta_{1}}}&	\qw	&\qw	&\qw		&\qw			&\qw	&\qw&\qw	&\qw							&\qw	&\qw	&		\\
 &  \qw&  \qw				&\qw	&\qw	&\qw		&\qw			&\qw	&\qw	&\qw	&\qw									&\ctrl{3}\qw& \ghost{\mathsmaller{\Delta_{1}}}	&\qw	&\qw	&\qw		&\qw			&\qw	&\qw	&\qw	&\qw							&\qw&	\qw	&		\\
  &  \qw	&  \qw			&\qw	&\ctrl{3}\qw&\qw&\qw			&\qw	&\ctrl{3}\qw &\qw&\qw						&\qw	& \ghost{\mathsmaller{\Delta_{1}}} 	&	\qw	&\ctrl{3}\qw &     \qw&\qw			&\qw  	&\ctrl{3}\qw&\qw&\qw				&\qw&\qw&			\\
   &  \ctrl{2}&  \qw			&\qw	&\qw	&\qw		&\qw			&\qw	&\qw	&\qw	& \qw									&\qw&	 \ghost{\mathsmaller{\Delta_{1}}}	&	\qw	&\qw	&\qw		&\qw			&\qw	&\qw	&\qw	&\qw							&\ctrl{2}&\qw	&	 \\
   &  \qw	&  \qw			&\qw	&\qw	&\qw	&    \ctrl{1}\qw	&\qw	&\qw	&	\qw	&    \qw								&\targ		& \ghost{\mathsmaller{\Delta_{1}}} & \qw	&\qw	&\qw	&    \ctrl{1}\qw	&\qw	&\qw	&	\qw	&    \qw						&\qw	&	\qw	&	\\
 &   \ctrl{1}	&  \qw			&\gate{\mathsmaller{A^{\dagger}}} &\targ&\gate{\mathsmaller{A^{\dagger}}}&\targ&\gate{\mathsmaller{A}} &\targ&\gate{\mathsmaller{A}}&\qw&\qw &\qw	&\gate{\mathsmaller{A^{\dagger}}} &\targ&\gate{\mathsmaller{A^{\dagger}}}&\targ&\gate{\mathsmaller{A}} &\targ&\gate{\mathsmaller{A}}&\qw&\ctrl{1}	&\qw	&		\\
 &\targ	&  \qw&\qw&\qw&\qw		&\qw	 &\qw	&\qw&\qw&\qw									&\qw &\qw		&	\qw&\qw&\qw		&\qw	 &\qw	&\qw&\qw&\qw	&\targ	&\qw		
\gategroup{3}{8}{6}{10}{.7em}{.} \gategroup{3}{14}{6}{16}{.7em}{.}  }
\]
where $A=R_{y}(\frac{\pi}{4})$.  The marked gates cancel each other
out, because they commute with the gates between them. The reset part
can be decomposed analogously.

\begin{table*}[!th] 
%% increase table row spacing, adjust to taste
\renewcommand{\arraystretch}{1.1}
\caption{\cnot{} counts and numbers of real parameters that can be
  introduced into a circuit by a specific gate, for various controlled
  gates.}
\label{tab:counts_overview}
\centering
\begin{ruledtabular}
\begin{tabular}{lccc}

Gate & Notation&  \cnot{} count (upper bound)&\# Real parameters   \\ \hline
UCG (up to a diagonal gate)&$  \Delta C^{\textnormal{u}}_{n-1}(U)$&$2^{n-1}-1$ \cite{10}&$2^n$\\ [0.04cm]   
Uniformly controlled rotation&$C^{\textnormal{u}}_{n-1}(R_z)$/$C^{\textnormal{u}}_{n-1}(R_y)$&$2^{n-1}$~\cite{diag,unif_rot}&$2^{n-1}$ \\ [0.04cm]  
Multi controlled unitary gate&$C_{n-1,n}(U)$& $16n^2-60n+42$ if $n\geqslant
3$ (Thm.~\ref{ThmC2})&4\\ [0.04cm]  
Multi controlled special unitary gate&$C_{n-1,n}(W)$& $28n-88$ if $n\geqslant 8$ is even (Thm.~\ref{ThmC})&3\\
 											     &($W \in SU(2)$)&$28n-92$ if $n\geqslant 8$ is odd (Thm.~\ref{ThmC})&\\ [0.04cm]  
Multi controlled \nt{} gate&$C_{k,n}(\sigma_x)$& $8k-6$ if $n\geqslant 5$, $k  \in \{  3,\dots,   \lceil  \frac{n}{2}   \rceil\}$ (Lemma~\ref{prop1c})&0
\end{tabular}
\end{ruledtabular}
\end{table*}

\renewcommand{\arraystretch}{1}

\begin{proof} [Proof of Lemma~\ref{prop1c}]
  First we apply Lemma~7.2 of~\cite{5} (a circuit diagram for the case
  $k=5$ and $n=9$ can be found in~\cite{5}). By similar arguments as
  used in the special case above, we introduce a corresponding
  diagonal gate for each Toffoli gate apart from the two that act on
  the least significant qubit (i.e., on the target qubit of the
  $C_{k,n}(\sigma_{x})$ gate).

  The required \cnot{} count for $C_{k,n}(\sigma_{x})$ is thus equal
  to twice that required for the reset part plus the number of
  \cnots{} needed to implement the Toffoli gates that form the first
  and last gate in the action part. By Lemma~\ref{lem1c2}, the two
  Toffoli gates can be decomposed using 12 \cnots. One reset part uses
  $N^{\text{reset}}_{C_{k,n}(\sigma_{x})}=4(k-3)+3$ \cnots. This leads
  to the claimed count.
\end{proof}

\begin{lem} [$C_{k,n}(\sigma_{x})$~{\cite[Lemma 7.3]{5}}]  \label{prop1c2}
Let $n \geqslant 5$ denote the total number of qubits considered. A
$C_{n-2,n}(\sigma_{x})$ gate can be decomposed into two
$C_{k,n}(\sigma_{x})$ and two $C_{n-k-1,n}(\sigma_{x})$ gates, where
$k \in \{ 2,3,\dots,n-3  \}$.\end{lem}

For example, the decomposition for $n=7$ and $k=4$ is shown in the following circuit diagram.
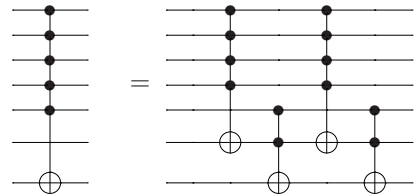
\[
\Qcircuit @C=1.0em @R=.7em {
  & \ctrl{1} & \qw &&&& 	\qw 	&\ctrl{1} 		&\qw 	&\ctrl{1}  &\qw 	&\qw 		 \\ 
   & \ctrl{1} & \qw &&&& 	\qw 	&\ctrl{1} 		&\qw 	&\ctrl{1}  &\qw 	&\qw 		 \\ 
   &   \ctrl{1} &\qw&&&&	\qw 	& \ctrl{1} 		&\qw  	 &\ctrl{1} &\qw 	&\qw 		 \\
    &\ctrl{1}  & \qw&&=&& 	\qw 	& \ctrl{2} 		&\qw  	& \ctrl{2}	&\qw &\qw 		 \\
    &\ctrl{2}& \qw&&&&	\qw 	&\qw			&\ctrl{1}	&\qw&\ctrl{1}&\qw  			\\
    &  \qw&\qw&&&&		\qw 	& \targ		&\ctrl{1}	 &\targ &\ctrl{1}	&\qw 		 \\
     & \targ& \qw&&&&		\qw  &\qw			&\targ	&\qw	&\targ&\qw 			 \\
}
\]

\begin{thm} [$C_{n-1,n}(U)$] \label{ThmC2} Let $n\geqslant 3$ and $U$
  be a single-qubit unitary. We can decompose a $C_{n-1,n}(U)$ gate
  using at most $16n^2-60n+42$ \cnots{}.
\end{thm}
\begin{proof}
The idea is contained in the following diagram in which $V$ is chosen
such that $V^2=U$ (see Lemma~7.5 of~\cite{5}).
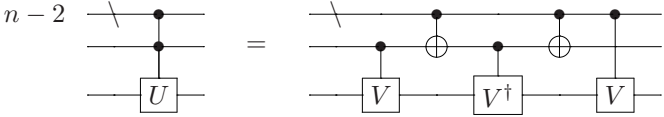
\[
\Qcircuit @C=1.0em @R=.7em {
%  & \ctrl{1} & \qw &&&& \qw 	&\ctrl{1} &\qw 	& \ctrl{1}	& \qw &\qw 		 \\ 
n-2&&&\backslash\qw   & \ctrl{2} & \qw &&&&&\backslash\qw & \qw 	&\ctrl{1} &\qw 	& \ctrl{1}	& \ctrl{2} &\qw 		 \\ %&&&&&&&&&&&\\
%      \vdots &&&&&& &	& \vdots	& & \vdots	& \vdots &	 \\
%    & \ctrl{1} & \qw&&&& \qw 	&\ctrl{1}	&\qw & \ctrl{1}	 & \ctrl{1} &\qw 		 \\
%    &  \ctrl{1} &\qw&&=&&\qw 	& \ctrl{1}&\qw 	 & \ctrl{1}	& \ctrl{1} &\qw 		 \\
%    &  \ctrl{1} &\qw &&&&\qw 	& \ctrl{2} &\qw & \ctrl{2}	&\ctrl{2}  &\qw 		 \\
    &&&\qw&    \ctrl{1}&\qw &&=&&&\qw&	\ctrl{1} &\targ	& \ctrl{1}&\targ	&\qw  	&\qw 	 \\
    &&&\qw&   \gate{U}&\qw &&&&&\qw&	\gate{V}&\qw 	& \gate{V^\dagger}&\qw 	& \gate{V}	&\qw 	 \\
}
\]
Using Lemma~\ref{cor5.3}, this gives the relation
$N_{C_{n-1,n}(U)}=N_{C_{n-2,n}(U)}+4+2N_{C_{n-2,n}(\sigma_x)}$. For
simplicity, we consider the $C_{n-2,n}(U)$ gate as a $C_{n-2,n-1}(U)$
gate. This will lead to an overcount in our final \cnot{} count. Using Lemma~\ref{prop1c2}  we have
$N_{C_{n-2,n}(\sigma_x)}= 2(N_{C_{\lceil
    n/2\rceil-1,n}(\sigma_x)}+N_{C_{\lfloor n/2\rfloor,n}(\sigma_x)})$
for $n\geqslant 5$ and hence, from Lemma~\ref{prop1c},
$N_{C_{n-2,n}(\sigma_x)}\leq 16n-40$ for $n\geqslant 5$. Note that
Lemma~\ref{lem1c2} implies that the same bound also holds for $n=4$
(although we know of a tighter bound in this case). Thus, we
wish to solve the recursion
$N_{C_{n-1,n}(U)}=N_{C_{n-2,n-1}(U)}+32n-76$. Noting that
$N_{C_{2,3}(U)}=6$ (cf.\ Remark~\ref{rmk:6}) we obtain the stated
count.
\end{proof}

Note that this count could be improved.  However, it turns out that
the case $W\in SU(2)$ is particularly useful.  In this case we make more
effort with the optimizations leading to the following.

\begin{thm} [$C_{n-1,n}(W)$, where $W\in SU(2)$] \label{ThmC}
  Let $n \geqslant 8$ and $W \in SU(2)$. We can decompose a
  $C_{n-1,n}(W)$ gate using at most $(28n-88)$ \cnots{} if $n$ is even
  and $(28n-92)$ \cnots{} if $n$ is odd.
\end{thm}

\begin{proof}  
To aid the proof, we provide illustrations for the case $n=8$. By Lemma 7.9 of~\cite{5} there exist quantum gates $A,B,C \in SU(2)$ such that we can decompose the $C_{n-1,n}(W)$ gate as follows.
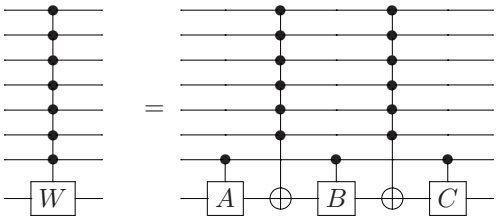
\[
\Qcircuit @C=1.0em @R=.7em {
  & \ctrl{1} & \qw &&&& \qw 	&\ctrl{1} &\qw 	& \ctrl{1}	& \qw &\qw 		 \\ 
   & \ctrl{1} & \qw &&&& \qw 	&\ctrl{1} &\qw 	& \ctrl{1}	& \qw &\qw 		 \\ 
   &   \ctrl{1} &\qw&&&&\qw 	& \ctrl{1}	&\qw & \ctrl{1}	& \qw&\qw 		 \\
    & \ctrl{1} & \qw&&&& \qw 	&\ctrl{1}	&\qw & \ctrl{1}	 & \qw &\qw 		 \\
    &  \ctrl{1} &\qw&&=&&\qw 	& \ctrl{1}&\qw 	 & \ctrl{1}	& \qw &\qw 		 \\
    &  \ctrl{1} &\qw &&&&\qw 	& \ctrl{2} &\qw & \ctrl{2}	&\qw  &\qw 		 \\
    &    \ctrl{1}&\qw &&&&	\ctrl{1} &\qw	& \ctrl{1}&\qw	&\ctrl{1}  	&\qw 	 \\
    &   \gate{W}&\qw &&&&	\gate{A}&\targ  	& \gate{B}&\targ 	& \gate{C}	&\qw 	 \\
}
\]

By Lemma~\ref{prop1c2} we can decompose the $C_{n-2,n}(\sigma_{x})$
gates using two $C_{k_1,n}(\sigma_{x})$ and two
$C_{k_2,n}(\sigma_{x})$ gates, where we set $k_2=\lceil n/2 \rceil$
and $k_1=n-k_2-1$. In our example $k_1=4$ and $k_2=3$:
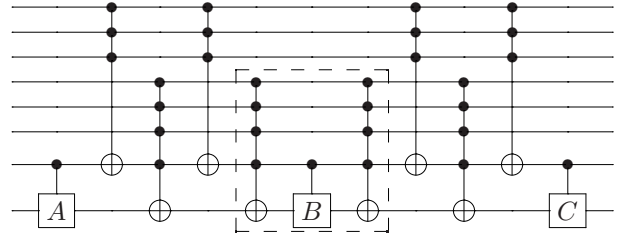
\[
\Qcircuit @C=1.0em @R=.7em {
& \qw&\ctrl{1}& \qw  	&\ctrl{1}& \qw  	&				\qw 	&	\qw&\ctrl{1}& \qw  	&\ctrl{1}& \qw&\qw   	 		 \\ 
& \qw&\ctrl{1}& 	\qw	&\ctrl{1}& 	\qw	&			\qw 	&	\qw&\ctrl{1}& 	\qw	&\ctrl{1}& 	\qw	&\qw 		 \\ 
  &\qw&\ctrl{4}& 	\qw	&\ctrl{4}& 	\qw	&				\qw & 	\qw	&\ctrl{4}& 	\qw	&\ctrl{4}& 	\qw	&\qw 		 \\
& \qw &\qw &	 \ctrl{1}	&\qw &	 \ctrl{1}	&		\qw & 	\ctrl{1}&\qw &	 \ctrl{1}	&\qw &\qw&\qw 	 		 \\
&\qw& \qw &	 \ctrl{1}	& \qw &	 \ctrl{1}	&		\qw 	 &	\ctrl{1}& \qw &	 \ctrl{1}	& \qw &\qw&\qw 		 \\
&\qw& \qw &	 \ctrl{1}	& \qw &	 \ctrl{1}	&		\qw &	\ctrl{1}& \qw &	 \ctrl{1}	& \qw &\qw&\qw 	 		 \\
 &\ctrl{1}&\targ&  \ctrl{1}	&\targ&  \ctrl{1}	&			 \ctrl{1}&	\ctrl{1}&\targ&  \ctrl{1}	&\targ&\ctrl{1}&\qw 	 		 \\
 &	\gate{A}& \qw&\targ &\qw&\targ &				 \gate{B}&	\targ	 & \qw&\targ &\qw&	 \gate{C}	&\qw
 \gategroup{4}{6}{8}{8}{.7em}{--}  
}
\]

\begin{figure*}[!t] 
\centering

\begin{equation*}
  \mathbf{\ket{\psi^{\textnormal{e}}}} = \begin{array}{c@{\!\!\!}l}
  \left[ \begin{array}[c]{ccccc}
    0\\
   \vdots \\
   0\\
    \hline
    c_{a^k_{s}} \\
     c_{a^k_{s}+1} \\
    \hline
    c_{a^k_{s }+2} \\
     c_{a^k_{s }+3} \\
      c_{a^k_{s } +4}\\
     c_{a^k_{s  } +5} \\
      \vdots \\
       c_{2^{n-s}-2}   \\
     c_{2^{n-s} -1 }  \\
  \end{array}  \right]

\end{array}
\phantom{a}
 \xrightarrow{\phantom{a} A \phantom{a}}  
  \mathbf{\ket{\psi'^{\textnormal{e}}}} = \begin{array}{c@{\!\!\!}l}
  \left[ \begin{array}[c]{ccccc}
    0\\
   \vdots \\
   0\\
    \hline
   c'_{a^k_{s+1}} \\
    0 \\
    \hline
    c'_{a^k_{s+1}+1} \\
     0   \\
        c'_{a^k_{s+1}+2} \\
    0   \\
      \vdots \\
      c'_{2^{n-(s+1)}-1}  \\
   0  \\
  \end{array}  \right]

\end{array}
 \phantom{a}, \phantom{aaaa}
  \mathbf{\ket{\psi^{\textnormal{e}}}} = \begin{array}{c@{\!\!\!}l}
  \left[ \begin{array}[c]{ccccc}
    0\\
   \vdots \\
   0\\
    \hline
   0 \\
    c_{a^k_{s} }\\
    \hline
     c_{a^k_{s}+1} \\
     c_{a^k_{s}+2 }  \\
      c_{a^k_{s} +3}  \\
   c_{a^k_{s} +4   }\\
      \vdots \\
       c_{2^{n-s}-2}    \\
     c_{2^{n-s}-1}    \\
  \end{array}  \right]
\end{array}
\phantom{a}
 \xrightarrow{\phantom{a} A \phantom{a}}  
  \mathbf{\ket{\psi'^{\textnormal{e}}}} = \begin{array}{c@{\!\!\!}l}
  \left[ \begin{array}[c]{ccccc}
    0\\
   \vdots \\
   0\\
    \hline
   0 \\
    c'_{a^k_{s+1}} \\
    \hline
    0  \\
    c'_{a^k_{s+1}+1}   \\
       0  \\
     c'_{a^k_{s+1}+2}     \\
      \vdots \\
     0  \\
    c'_{2^{n-(s+1)}-1}    \\
  \end{array}  \right]
  
\end{array}
\end{equation*}

\caption{Using a quantum gate $A$ to disentangle the $(n-s)$th qubit into the state $k_s=0$ or $k_s=1$ respectively.}
\label{fig:disentanglement_matrix_not}
\end{figure*}

Since the $C_{n-2,n}(\sigma_{x})$ gate is its own inverse, we can use
the inverted decomposition scheme to decompose the second
$C_{n-2,n}(\sigma_{x})$ gate.  We can decompose the gates
$C_{k_1,n}(\sigma_{x})$ and $C_{k_2,n}(\sigma_{x})$ using
Lemma~\ref{prop1c}. Note that this works for all $n\geqslant 8$, since
$3 \leqslant k_1,k_2\leqslant \lceil n/2 \rceil$. We can lower the
\cnot{} count with some technical tricks. As in the proof of
Corollary~7.4 of~\cite{5} we can decompose all Toffoli gates not
acting on the least significant qubit up to diagonal gates. This can
be seen by reversing the decomposition scheme of Lemma~\ref{prop1c}
for the second and fourth $C_{k_1,n}(\sigma_{x})$ gate and using
Lemma~\ref{commutation}.  Therefore, using the same technique as in
Lemma~\ref{prop1c}, but implementing all Toffoli gates up to diagonal
gates, we can decompose each of the $C_{k_1,n}(\sigma_{x})$ gates
using $N_{C_{k_1,n}(\sigma_{x})}-2\cdot6+2\cdot2=8k_1-14$ \cnots.

Now consider the marked part of the last circuit. By
Lemma~\ref{prop1c} this can be decomposed using
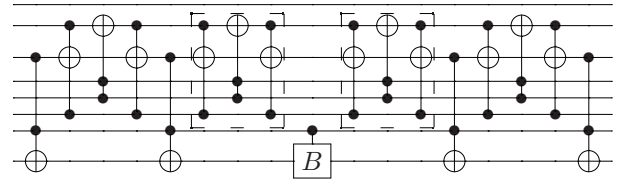
\[
\Qcircuit @C=0.45em @R=.38em {
& \qw& \qw& \qw& \qw& \qw& \qw& \qw& \qw& \qw& \qw& \qw& \qw& \qw& \qw& \qw& \qw& \qw & \qw \\ 
& \qw&\ctrl{1}&\targ&\ctrl{1}& \qw&\ctrl{1}&\targ&\ctrl{1}&\qw&\ctrl{1}&\targ&\ctrl{1}& \qw&\ctrl{1}&\targ&\ctrl{1}& \qw& \qw  \\ 
& \ctrl{4}&\targ& \qw&\targ&\ctrl{4}&\targ& \qw&\targ&\qw&\targ& \qw&\targ& \ctrl{4}&\targ& \qw&\targ&\ctrl{4}& \qw  \\ 
& \qw& \qw& \ctrl{-2}& \qw&\qw&\qw& \ctrl{-2}&\qw&\qw&\qw& \ctrl{-2}&\qw& \qw& \qw& \ctrl{-2}& \qw&\qw& \qw  \\ 
& \qw& \qw& \ctrl{-1}& \qw&\qw& \qw& \ctrl{-1}& \qw&\qw& \qw& \ctrl{-1}& \qw& \qw& \qw& \ctrl{-1}& \qw&\qw & \qw \\ 
& \qw& \ctrl{-3}& \qw& \ctrl{-3}& \qw&\ctrl{-3}& \qw& \ctrl{-3} & \qw &\ctrl{-3}& \qw& \ctrl{-3}& \qw& \ctrl{-3}& \qw& \ctrl{-3}& \qw& \qw \\ 
& \ctrl{1}& \qw& \qw& \qw& \ctrl{1}& \qw& \qw& \qw&\ctrl{1}& \qw& \qw& \qw& \ctrl{1}& \qw& \qw& \qw& \ctrl{1}& \qw \\ 
&\targ& \qw& \qw& \qw&\targ& \qw& \qw& \qw& \gate{B}& \qw& \qw& \qw&\targ& \qw& \qw& \qw&\targ& \qw 
\gategroup{2}{7}{6}{9}{.7em}{--} \gategroup{2}{11}{6}{13}{.7em}{--}
}
\]
where, to simplify, we have not explicitly illustrated the diagonal
gates. The two reset parts commute with the controlled $B$ gate, since
they don't act on the two least significant qubits, and cancel out.
Therefore each of the marked $C_{k_2,n}(\sigma_{x})$ gates uses
$N_{C_{k_2,n}(\sigma_{x})}-N^{\text{reset}}_{C_{k_2,n}(\sigma_{x})}=4k_2+3$
\cnots.  We decompose the other two $C_{k_2,n}(\sigma_{x})$ gates
exactly as in Lemma~\ref{prop1c}. Using Lemma~\ref{cor5.3} for the
three single controlled gates then leads to the claimed \cnot{} count.
\end{proof}

\subsection{Overview of  \cnot{} counts for controlled gates} \label{sec:counts_overview}

We summarize \cnot{} counts for some commonly-used uniformly and not
uniformly controlled gates in Table~\ref{tab:counts_overview}. Note
that implementing a uniformly controlled $C^{\textnormal{u}}_{n-1}(U)$
gate up to a diagonal gate $\Delta$ means that we implement
$ \Delta C^{\textnormal{u}}_{n-1}(U)$, for some diagonal gate
$\Delta$. The number of real parameters required to
specify a particular gate is shown in the final column and follows
from Lemma~\ref{ZYZ} and the block diagonal form of the uniformly
controlled gates (see also the argument used to derive the lower bound
for isometries in Section~\ref{sec:lower_bound}).  For example, a
$C^{\textnormal{u}}_{n-1}(U)$ gate is described by $2^{n-1}$
$(2\times2)$-unitaries. By Lemma~\ref{ZYZ} this corresponds to $4
\cdot 2^{n-1}$ real parameters. Since a diagonal gate $\Delta$ on $n$
qubits is described by $2^n$ real parameters, a
$\Delta C^{\textnormal{u}}_{n-1}(U)$ gate is described by
$4\cdot2^{n-1}-2^n=2^n$ real parameters.

\subsection{Rigorous proof of the decomposition scheme described in Section~\ref{sec:isounif} and exact \cnot{} count} \label{sec:app_isounif}

We begin this section by introducing some additional notation. For $m' \in \mathbb{N}$ and $k \in  \{0,1,\dots , 2^{m'}-1\}$ we use the notation: $k=[k_{m'-1},k_{m'-2},\ldots,k_0]:=\sum_{i=0}^{m'-1}k_{i}2^{i}$, i.e., $\{k_i\}$ are the binary digits of $k$. For $s\in\mathbb{N}_0$ we define $a^k_{s},b^k_{s} \in \mathbb{N}_0$ by $k=a^k_{s}2^s+b^k_{s}$,  such that $a^k_{s}$ is maximal. For $s\in\{1,2,\dots,n'-1\}$, where $n'\in\mathbb{N}_{\geqslant 2}$ and $n'\geqslant m'$, we can also write $a^k_{s}=[k_{n'-1},k_{n'-2},\ldots,k_s]$ and $b^k_{s}=[k_{s-1},k_{s-2},\ldots,k_0]$. 

We now consider an elementary step in the decomposition scheme. Let
$n\in\mathbb{N}_{\geqslant 2}$, $m\in\mathbb{N}$ with $n\geqslant m$,
$k \in \{1,2,\dots ,2^n-1\}$ and $s \in \{0,1,\dots ,n-2
\}$. Furthermore suppose $\ket{\psi}$ is an $n$-qubit state of the
form
    \begin{equation}  \label{eq:lem1:psi_init_form_s}
\ket{\psi}=\left(\sum_{l=a^k_s}^{2^{n-s}-1} c_l\ket{l}\right) \otimes \ket{k_{s-1}k_{s-2}\dots k_{0}},
\end{equation}
where $c_l\in \mathbb{C}$ for all $l \in\{a^k_s,a^k_s+1,\dots,
2^{n-s}-1 \}$. Since it is clear from the context that, e.g., $\ket{l}
\in \mathcal{H}_{n-s}$, we shorten the notation and write $\ket{l}$
instead of $\ket{l}_{n-s}$.

[Note that we use the following convention: If $s-1<0$, we mean that
the part $ \ket{k_{s-1}k_{s-2}\dots k_{0}}$
in equation~(\ref{eq:lem1:psi_init_form_s}) does not exist, i.e., for $s=0$ the
statement of  equation~(\ref{eq:lem1:psi_init_form_s}) is:
$\ket{\psi}=\sum_{l=a^k_0}^{2^{n}-1} c_l\ket{l}$. Analogously,
$I^{\otimes 0}$ means that no such part exists in the considered
expression. Similarly we set $\{n_s,\dots,n_e \} =\emptyset$ if
$n_e<n_s$.]

% Figure for Lemma "Implementing one column of an isometry
 
 \begin{figure*}[h!tp]

\centering
\[
\Qcircuit @C=1em @R=.45em {
 &\qw& \gate{\mathsmaller{*}} \qwx[1] \qw  & \gate{} \qwx[1] \qw 	&\multigate{7}{\mathsmaller{\Delta_{0}}}			& \gate{\mathsmaller{*}} \qwx[1] \qw  & \gate{} \qwx[1] \qw &\multigate{6}{\mathsmaller{\Delta_{1}}}	&\qw{.}&{.}&{.}   		&\gate{\mathsmaller{*}} \qwx[1] \qw&    \gate{} \qwx[1] \qw&\multigate{1}{\mathsmaller{\Delta_{n-2}}}&\gate{\mathsmaller{{U}_{n-1}}} &\gate{U^{ \textnormal{u} }_{n-1}}&\gate{\mathsmaller{\Delta_{n-1}}} \qw	\\
& \qw & \gate{\mathsmaller{*}} \qwx[1] \qw  & \gate{} \qwx[1] \qw	&\ghost{\mathsmaller{\Delta_{0}}}			& \gate{\mathsmaller{*}} \qwx[1] \qw  & \gate{} \qwx[1] \qw  	&\ghost{\mathsmaller{\Delta_{1}}}		&\qw{.}&{.}&{.}			&\gate{\mathsmaller{{U}_{n-2}}} \qwx[1] \qw&	\gate{\mathsmaller{U^{ \textnormal{u} }_{n-2}}}&\ghost{\mathsmaller{\Delta_{n-2}}}  &\gate{\mathsmaller{*}} \qwx[-1] \qwx[1]& \qw& \qw\qw	\\
&&{.}&		{.}				&		&{.}														&{.}	&											&&	&				&     {.}     	 	&  	&  &{.}  & & & & 	\\
&&{.}& 			{.}		&	&{.}															&{.}	&											&&	&				&{.}			&   	&           &    {.} && &   \\
&&{.}&				{.}		&	&{.}															&{.}&												&&	&				&{.}			&   	& &{.} && &  \\
& \qw& \gate{\mathsmaller{*}} \qwx[-1]\qwx[1] \qw&\gate{} \qwx[1]  \qwx[-1]  \qw&\ghost{\mathsmaller{\Delta_{0}}}			&\gate{\mathsmaller{*}} \qwx[-1] \qw  & \gate{} \qwx[1]\qwx[-1] \qw 		&\ghost{\mathsmaller{\Delta_{1}}}	&\qw{.}&{.}&{.}		&\gate{\mathsmaller{*}} \qwx[-1] \qw&\qw	&\qw	&\gate{\mathsmaller{*}} \qwx[-1] & \qw& \qw\qw\\
& \qw& \gate{\mathsmaller{*}} \qwx[1] \qw&\gate{} \qwx[1]  \qw 	&	\ghost{\mathsmaller{\Delta_{0}}}	&\gate{\mathsmaller{{U}_{1}}} \qwx[-1] \qw&\gate{\mathsmaller{U^{ \textnormal{u} }_{1}}} \qwx[-1] \qw& \ghost{\mathsmaller{\Delta_{1}}}	&\qw{.}&{.}&{.}			&\gate{\mathsmaller{*}} \qwx[-1] \qw&\qw&\qw	&\gate{\mathsmaller{*}} \qwx[-1] & \qw& \qw\qw	\\
& \qw&\gate{\mathsmaller{{U}_{0}}} \qw&	\gate{\mathsmaller{U^{ \textnormal{u} }_{0}}}  \qw 	&	\ghost{\mathsmaller{\Delta_{0}}}	&	\gate{\mathsmaller{*}} \qwx[-1] \qw  &\qw&\qw&		\qw															{.}&{.}&{.}			&\gate{\mathsmaller{*}} \qwx[-1] \qw&	\qw&\qw	&\gate{\mathsmaller{*}} \qwx[-1]& \qw& \qw \qw \\
}
\]
\caption{Decomposition scheme of a quantum gate $G_k$. The notation
  ``$*$'' surrounded by the square signifies either a control on one
  or on zero.}
\label{fig:decG_k}

\end{figure*}

\begin{lem} \label{unif_disent.} Take $\ket{\psi^{\textnormal{e}}} :=
  \sum_{l=a^k_s}^{2^{n-s}-1} c_l\ket{l}$, where ``\textnormal{e}"
  stands for entangled and assume that
\begin{equation}
      \label{eq:lem_add_cond.}
c_{2a^k_{s+1}+1}=0 \text{ if } k_s=0 \text{ and } b^k_{s+1}\neq 0.
\end{equation}

There exists a UCG $A:=C_{n-1-s}^{\textnormal{u}}(U)$ of the form
 \begin{equation} 
  \label{eq:lem1:form_A_s}
A=\sum_{l=0}^{2^{n-1-s}-1} \ketbra{l}{l} \otimes U_l \otimes I^{ \otimes s} ,
\end{equation}
such that $\ket{\psi'}:=A\ket{\psi}$ has the form
     \begin{equation}
     \label{eq:lem1:form_psi'_s}
	\ket{\psi'}=\left(\sum_{l=a^k_{s+1}}^{2^{n-(s+1)}-1} c'_l\ket{l}\right) \otimes  \ket{k_{s}k_{s-1}\dots k_{0}}, 
\end{equation}
where $c'_l\in \mathbb{C}$ for all $l
\in\{a^k_{s+1},a^k_{s+1}+1,\dots, 2^{n-(s+1)}-1 \}$. Additionally, $A$
has the property that
\begin{equation}
  \label{eq:lem1:action_of_A_on_basis}
  A\ket{i}=\ket{i} \text{ for all } i\in\{0,1,\ldots,k-1\}.
\end{equation}

\end{lem}
\begin{proof} 
The following proof depends on whether $k_s=0$ or $k_s=1$. In the case  $k_s=0$ we has also to distinguish between the cases $b^k_{s+1}=0$ and $b^k_{s+1}\neq 0$. The reader might find it useful to read the proof first considering only the case $k_s=1$ (and therefore $b^k_{s+1}\neq0$).\newline
Considering blocks of two elements,  there exist two possible forms of $\ket{\psi^{\textnormal{e}}}$, depending on whether $k_s=0$ or $k_s=1$. If $k_{s}=0$, then $a_s^k=2 a^k_{s+1}$ is even and therefore $\ket{\psi^{\textnormal{e}}}$ begins with an even number of zeros (assuming $c_{a_s^k}\neq 0$). If $k_s=1$, then $a_s^k=2 a^k_{s+1}+1$ is odd and  $\ket{\psi^{\textnormal{e}}}$ begins with an odd number of zeros (see Fig.~\ref{fig:disentanglement_matrix_not}). By equation~(\ref{eq:lem1:form_A_s}) the quantum gate $A$ leaves the $s$ lower significant qubits invariant and we can write:  $A \ket{\psi}= \left(\sum_{l=0}^{2^{n-s}-1} {c'}_{l}^{\textnormal{e}} \ket{l}\right) \otimes  \ket{k_{s-1}k_{s-2}\dots k_{0}}$ for some coefficients ${c'}_l^{\textnormal{e}} \in \mathbb{C}$. We define  $\ket{\psi'^{\textnormal{e}}}:= \sum_{l=0}^{2^{n-s}-1} {c'}_{l}^{\textnormal{e}} \ket{l}$. We want to find a gate $A$, such that  for $l'\in\{0,1,\dots,2^{n-s-1}-1\}$:  ${c'}_{2l'+1}^{\textnormal{e}}=0$ if $k_s=0$, and ${c'}_{2l'}^{\textnormal{e}}=0$ if $k_s=1$, i.e., we want to disentangle the $(n-s)$th qubit into the state $\ket{k_s}$.

We now determine the UCG $A$. To ensure that $A$
fulfils equation~(\ref{eq:lem1:action_of_A_on_basis}) we set:
\begin{subnumcases}{ \label{eq:lem1:setu} \!\!\!\!\!\!\!\! U_l=\! } 
\! I \text{ for } l \in \{0,1,\dots, a^k_{s+1}\} \text{ if }b^k_{s+1} \neq 0 \label{eq:lem1:setu1.1},\\
\! I \text{ for } l \in \{0,1,\dots ,a^k_{s+1}-1\} \text{ if }b^k_{s+1}=0. \text{~ ~ }\label{eq:lem1:setu1.2}
\end{subnumcases}

If the gate $A$ is not already fully specified by equation~(\ref{eq:lem1:setu}), we use Lemma~\ref{rotate_to_basis_state} to determine the gates $U_l$ for $l \in \{a^k_{s+1}+1,a^k_{s+1}+2,\dots,2^{n-1-s}-1\}$ if $b^k_{s+1}\neq 0$ and for $l \in \{a^k_{s+1},a^k_{s+1}+1,\dots,2^{n-1-s}-1\}$ if $b^k_{s+1}=0$:

\begin{subnumcases}{U_l \left(\begin{array}{c} c_{2l} \\ c_{2l+1} \end{array}\right)=} 
r  \left(\begin{array}{c} 1 \\ 0 \end{array}\right) \text{ if }k_s=0,\label{eq:lem1:setu2.1}\\
r \left(\begin{array}{c} 0 \\  1\end{array}\right)   \text{ if }k_s=1,\label{eq:lem1:setu2.2}
\end{subnumcases}
where $r \in \mathbb{R}$. [Note that if $b^k_{s+1}=0$ and $l=a_{s+1}^k$, the gate $A$ acts trivially on $\ket{i}$ for all  $i\in\{0,1,\ldots,k-1\}$, because of the form of the gate $A$ and since $a_{s+1}^k>a_{s+1}^{i}$ for all  $i\in\{0,1,\ldots,k-1\}$ in the considered case.]\\
With this choice of the gate $A$ we conclude: For all $l \in
\{a^k_{s+1}+1,a^k_{s+1}+2,\dots,2^{n-1-s}-1\}$ we have
${c'}^{\textnormal{e}}_{2l+1}=0$ if $k_s=0$ and
${c'}^{\textnormal{e}}_{2l}=0$ if $k_s=1$. Because of the initial form
of $\ket{\psi^{\textnormal{e}}}$ and the construction of the gate $A$
we conclude further that ${c'}_{l'}^{\textnormal{e}}=0$ for $l' \in
\{0,1, \dots, 2a^k_{s+1}-1\}$. It remains to consider the two coefficients ${c'}_{2a^k_{s+1}}^{\textnormal{e}}$ and ${c'}_{2a^k_{s+1}+1}^{\textnormal{e}}$. \\
If $k_s=0$ and $b^k_{s+1}=0$, then we can zero the coefficient
$c_{2a^k_{s+1}+1}$ with the gate $A$ (see
equation~(\ref{eq:lem1:setu2.1})). In the case $k_s=0$ and
$b^k_{s+1}\neq0$ the coefficient $c_{2a^k_{s+1}+1}$ is zero by
assumption and we act trivially on it with the gate $A$ by
equation~(\ref{eq:lem1:setu1.1}).
If  $k_s=1$,  then ${c'}_{2a^k_{s+1}}^{\textnormal{e}}=0$ because the corresponding entry in $\ket{\psi^{\textnormal{e}}}$ is initially zero by  equation~(\ref{eq:lem1:psi_init_form_s}) and $A$ acts trivially on it by equation~(\ref{eq:lem1:setu1.1}).\\
So in all cases we can write
$\ket{\psi'^{\textnormal{e}}}=\left(\sum_{l=a^k_{s+1}}^{2^{n-(s+1)}-1}
  c'_l\ket{l}\right) \otimes \ket{k_s}$, for some $c'_l \in
\mathbb{C}$ (see
Fig.~\ref{fig:disentanglement_matrix_not}). Therefore, $A\ket{\psi}$
is of the desired form~(\ref{eq:lem1:form_psi'_s}) and by construction
$A$ satisfies
equation~(\ref{eq:lem1:action_of_A_on_basis}).\end{proof}

\begin{lem} \label{zero_entry}
Let  $k \in  \{1,2,\dots ,2^n-1\}$ and   $s \in  \{0,1,\dots ,n-1 \}$
be such that  $k_s=0$ and $b^k_{s+1}\neq 0$. Let $\ket{\psi}$ be an $n$-qubit  state of the form  equation~\eqref{eq:lem1:psi_init_form_s}. Then there exist a MCG $B:=C_{n-1}(U)$, whose non
trivial part is of the form $\ketbra{K_1}{K_1} \otimes U \otimes
\ketbra{K_0}{K_0}$, where $K_1=[k_{n-1},k_{n-2},\dots,k_{s+1}]$ and
$K_0=[k_{s-1},k_{s-2},\dots,k_{0}]$, such that we can write
\begin{equation}
     \label{eq:lem2:form_psi'2}
	\ket{\psi'}:=B \ket{\psi}=\left(\sum_{l=a^k_s}^{2^{n-s}-1} c'_l\ket{l}\right) \otimes  \ket{k_{s-1}k_{s-2}\dots k_{0}},
\end{equation}
where $c'_l\in \mathbb{C}$ for all $l\in\{a^k_s,a^k_s+1,\dots,
2^{n-s}-1 \}$ and $c'_{2a^k_{s+1}+1}=0$. In addition, $B$ leaves the
first $k$ basis states invariant
\begin{equation}
      \label{eq:lem2:action_ofB_on_basis}
	B\ket{i}=\ket{i} \text{ for } i\in\{0,\ldots,k-1\} .
\end{equation}

\end{lem}

\begin{proof} 
Since $k_s=0$ the condition~(\ref{eq:lem2:action_ofB_on_basis}) is satisfied by construction of the gate $B$. We define the gate $U$ with Lemma~\ref{rotate_to_basis_state} such that
 \begin{equation}
 \label{eq:lem2:setB}	
U \left(\begin{array}{c}  c_{2a^k_{s+1}} \\ c_{2a^k_{s+1}+1} \end{array}\right)=r  \left(\begin{array}{c} 1 \\ 0 \end{array}\right),
\end{equation}
where $r \in \mathbb{R}$.\end{proof}

\begin{lem} [One column of an isometry]  \label{prop_iso}
Let $k \in  \{1,2, \dots, 2^n-1\}$. Let $\ket{\psi} \in \mathcal{H}_{n}$  be an $n$-qubit state such that $\braket{i}{\psi}=0$ for $i \in  \{0,1, \dots ,k-1\}$.  There exist a quantum gate $G_k$ with the following properties:
\begin{eqnarray}
\label{eq:prop1_Gk}
	G_{k}\ket{i}&=&e^{\I \varphi_{i}}\ket{i},\text{ }i \in \{0,1,\dots, k-1\},\\
\label{eq:prop2_Gk}
	G_{k}\ket{\psi}&=&e^{\I \varphi_{k}}\ket{k},	
\end{eqnarray}
where $\varphi_{i}\in\mathbb{R}$ for all $i \in \{0,1, \dots,k\}$. \\
\end{lem} 

\begin{proof} 
We claim that we can implement the operator  $G_{k}$ with a circuit of the form as shown in Fig.~\ref{fig:decG_k}. 

[Note that we have interchanged the order of the MCGs and the UCGs compared with Section~\ref{sec:isounif}. We are allowed to do this, since the gates commute by their construction.]

The structure of this decomposition is based on the idea used for
state preparation in~\cite{10}. The
diagonal gates in $\{\Delta_{i}\}_{i \in \{0,1,\dots,n-1\}}$ are
present so we can use the efficient decomposition of the UCGs up to
diagonal gates in~\cite{10}. Note that we never use the MCG
$C_{n-1}(U_0)$, since we can absorb it into the UCG $C^{
  \textnormal{u} }_{n-1}(U^{ \textnormal{u} }_0)$.  Formally we write:
\begin{equation*}
\label{eq:partition_of_G_k}
G_k=\prod_{s=0}^{n-1} O_s:=\prod_{s=0}^{n-1}\left( \Delta_{s} \otimes I^{\otimes s}\right)    C^{ \textnormal{u} }_{n-1-s}(U_{s}^{ \textnormal{u} })  C_{n-1}(U_{s}).
\end{equation*}

To keep the notation simple, we don't write down which of the $n$
qubits are the control/target qubits. The target qubit of the
controlled gates with lower index $s$ is the $(n-s)$th qubit. We
consider all controlled gates as $n$ qubit gates. If there are free
qubits, i.e., qubits that are neither controlled nor acted on, they
are the least significant ones.

We use  Lemma~\ref{unif_disent.} recursively to disentangle one qubit
after another starting from the state $\ket{\psi}$. More formally: We
define the state  $\ket{\psi_{s}}:=\prod_{s'=0}^{s-1} O_{s'}
\ket{\psi} $ for $s \in \{1,2,\dots,n\}$ and we set
$\ket{\psi_{0}}:=\ket{\psi}$.  To determine the gate $C^{
  \textnormal{u} }_{n-1-s}(U_{s}^{ \textnormal{u} })$ for $s\in
\{0,1,\dots,n-2\}$ we  apply Lemma~\ref{unif_disent.} on the state
$\ket{\psi'_{s}}:=C_{n-1}(U_{s}) \ket{\psi_{s}}$.  If $k_s=0$ and
$b^k_{s+1}\neq 0$, $ \ket{\psi_{s}}$ does not satisfies the
condition~(\ref{eq:lem_add_cond.}) for Lemma~\ref{unif_disent.} in
general. In this case we can determine the MCG $C_{n-1}(U_{s})$ by
Lemma~\ref{zero_entry}, such that  $ \ket{{\psi'}_s}$
satisfies the condition~(\ref{eq:lem_add_cond.}). In all other cases we set
$C_{n-1}(U_{s})=I$. Note, that the diagonal gate   $\Delta_{s} \otimes
I^{\otimes s}$ leaves the form of the state $C^{ \textnormal{u}
}_{n-1-s}(U_{s}^{ \textnormal{u} }) \ket{{\psi'}_s}$ invariant up to
phase shifts.

In  the case $s=n-1$  we have $b^k_{n}\neq 0$ and so either the most
significant qubit is initially disentangled $(k_{n-1}=1)$ or can be
disentangled with the MCG  $C_{n-1}(U_{n-1})$, determined by
Lemma~\ref{zero_entry} ($k_{n-1}=0$). Therefore we set $C^{
  \textnormal{u} }_{{0}}(U_{n-1}^{ \textnormal{u} })=I$ and
$\Delta_{n-1}=I$.

By construction, the operators $O_s$ leave the states $\{\ket{i}\}_{i \in \{0,1,\dots,k-1\}}$ invariant (up to phase shifts caused by the diagonal gates).\end{proof}

\begin{lem} [\cnot{} count for one column]  \label{prop_iso_CNOT_count}
Let $k \in  \{1,2, \dots, 2^n-1\}$. We can decompose a quantum gate $G_k$, which is of the form as describe in Lemma~\ref{prop_iso}, using at most $((2^n-n-1)+Q^k(n) N_{C_{n-1}(U)})$ \cnots, where $Q^k(n):=| \{s: k_s=0 \wedge b_{s+1}^k \neq 0\text{, } s \in \{0,1,\dots, n-1 \}  \}|$ and $N_{C_{n-1}(U)}$ denotes the number of \cnots{} used to decompose an $C_{n-1}(U)$ gate.
\end{lem} 

\begin{proof} 
To decompose the quantum gate $G_k$ we use the decomposition scheme described in the proof of  Lemma~\ref{prop_iso}.  The number of \cnots{} used to decompose the UCGs (together with the diagonal gates) give a count of $\Sigma_{s=0}^{n-1}(2^{n-1-s}-1)=2^{n}-n-1$ \cnots{}~\cite{10}. By the construction of the proof of  Lemma~\ref{prop_iso} we conclude, that the quantity of MCGs used for the decomposition of $G_k$ is  at most $Q^k(n)$. We add the number of \cnots{} used to decompose $Q^k(n)$ MCGs to the \cnot{} count used to decompose the UCGs and get the claimed count.\end{proof}

\begin{cor} \label{cor_total_number_mcg} The number of MCGs $Q(m,n)$
  used to decompose all operators in $\{G_i \}_{i \in
    \{1,2,\dots,2^{m}-1\}}$ using the decomposition scheme as in
  the proof of Lemma~\ref{prop_iso}, is given by:
\begin{equation}
\label{eq:total_mc_gates}
Q(m,n)=2^m \left(n-\frac{m}{2}-1\right)-n+m+1.
\end{equation}
\end{cor}

\begin{proof} 
We define the indicator function $I(k,s)$ by:
\begin{subnumcases}{I(k,s):=} 
1 \text{ if }k_s=0 \wedge b_{s+1}^k \neq 0 \label{indication_function1},\\
0   \text{ otherwise. } \label{indication_function2}
\end{subnumcases}

In other words
$I(k,s)=\delta_{k_s,0}(1-\delta_{b_{s+1}^k,0})=\delta_{k_s,0}-\delta_{b_{s+1}^k,0}$,
since $b_{s+1}^k=0$ implies $k_s=0$.  Now we can write $Q^k(n)=
\sum_{s=0}^{n-1} I(k,s)$. By Lemma~\ref{prop_iso_CNOT_count}:
\begin{equation}
\label{eq:total_mc_gates_calculatet_with_prop_4}
Q(m,n)=\sum_{k=1}^{2^m-1} Q^k(n)=\sum_{s=0}^{n-1} Q_s(m),
\end{equation}
where $Q_s(m):= \sum_{k=1}^{2^m-1}  I(k,s)$ denotes the number of MCGs acting on the $(n-s)$th qubit used to decompose all the gates in $\{G_i   \}_{i \in \{1,2,\dots,2^{m}-1\}}$.
If $m \leqslant s \leqslant n-1$ we have:
\begin{equation}
\label{eq:Q_s2}
Q_s(m)= \sum_{k=1}^{2^m-1}  I(k,s)=2^m-1,
\end{equation}
since $ I(k,s)=1$ for the whole index range. 
If $0 \leqslant s \leqslant m-1$ we include $k=0$ into the index range to simplify the combinatorial idea behind the following calculation:
\begin{equation}
\label{eq:Q_s1}
Q_s(m)= \sum_{k=0}^{2^m-1} \delta_{k_s,0}- \delta_{k \text{ mod }2^{s+1},0} =2^{m-1}- 2^{m-s-1}.
\end{equation}
Here we have used that $\delta_{b_{s+1}^k,0}= \delta_{k
  \text{mod}2^{s+1},0} $ by definition of $b_{s+1}^k$.  Plugging
everything into equation~(\ref{eq:total_mc_gates_calculatet_with_prop_4}), we
get the claimed count.\end{proof}

\begin{lem} [Column-by-column decomposition] \label{isocc} Let $V$ be
  an $m$ to $n$ isometry, described by a $2^n\times2^m$ matrix, and
  $I_{2^n\times2^m}$ denote the first $2^m$ columns of the
  $2^n\times2^{n}$ identity matrix. There exist quantum gates
  $G_{1},G_{2},\dots,G_{2^m-1}$ of the same form as in
  Lemma~\ref{prop_iso}, as well as a quantum gate $G_0$, which
  satisfies equation~(\ref{eq:prop2_Gk}) for an arbitrary $n$-qubit
  state $\ket{\psi}$, and a diagonal gate $ \Delta$ acting on $m$
  qubits, such that
\begin{equation}
\label{eq:iso3}
	 G_{0}^{\dagger}G_{1}^{\dagger} \dots G_{2^m-1}^{\dagger}  \left(I^{\otimes(n-m)}\otimes \Delta^{\dagger}\right) I_{2^n\times2^m}=V.
\end{equation}
\end{lem}
\begin{proof} 
Assume that we know a decomposition of a quantum gate  $G$ into
one-qubit and \cnot{} gates. We can inverse its order and take the
conjugate transpose of the one-qubit gates to get a decomposition of
$G^{\dagger}$, since a \cnot{} gate is inverse to itself. In
particular, $G^{\dagger}$ and $G$ can be implemented using the same number of \cnots{}. This allows us to replace equation~(\ref{eq:iso3}) by
\begin{equation}
\label{eq:iso_dec_inv}
	I_{2^n\times2^m}= (I^{\otimes(n-m)}\otimes\Delta)G_{2^m-1}G_{2^m-2}\dots G_{0}V.
\end{equation}
By definition of the gate $G_{0}$, we can choose it such that
$G_{0}V\ket{0}_m=e^{\I \varphi_0^0}\ket{0}_n$, where $\varphi_0^0 \in
\mathbb{R}$. Since the columns of an isometry are orthonormal and
$G_{0}$ is unitary, the columns of $G_{0}V$ are also orthonormal (for
example, $|\!\phantom._n\!\bra{0}G_{0}V\ket{0}_m|=1$ implies that
$\phantom._n\!\bra{0}G_{0}V\ket{1}_m=0$). We can therefore choose
$G_1$, such that $G_1G_{0}V\ket{1}_m=e^{\I \varphi_1^1}\ket{1}_n$,
where $\varphi_1^1 \in \mathbb{R}$ . By definition of $G_{1}$,
$G_{1}G_{0}V\ket{0}_m=e^{\I \varphi_0^1}\ket{0}_n$, where $\varphi_0^1
\in \mathbb{R}$. If we continue this procedure, we get
$G_{2^m-1}G_{2^m-2}\dots G_{0}V\ket{i}_m=e^{\I
  \varphi_{i}^{2^m-1}}\ket{i}_n$ for $i \in \{0,1,\dots,2^{m}-1 \}$,
where $\varphi_{i}^{2^m-1} \in \mathbb{R}$. We clear up the phases
with a diagonal gate $\Delta$ acting on the $m$ lower significant
qubits, such that
$(I^{\otimes(n-m)}\otimes\Delta)G_{2^m-1}G_{2^m-2}\dots
G_{0}V\ket{i}_m=\ket{i}_n$ for $i \in \{0,1,\dots,2^{m}-1 \}$, which
is equivalent to equation~(\ref{eq:iso_dec_inv}).\end{proof}

\begin{thm} [\cnot{} count for an isometry] \label{isocc_CNOT}
Let $m$ and $n$ be natural numbers with $n \geqslant 8$
  and $V$ be an isometry from $m$ qubits to $n$ qubits. There exists a
  decomposition of $V$ in terms of single-qubit gates and \cnots{} 
  such that the number of \cnot{} gates required satisfies
\begin{equation}
\label{eq:iso_CNOT_general_count}
N_{\mathrm{iso}}(m,n) \leqslant N_{\textrm{SP}}(n)+N_G(m,n)+N_{ \Delta}(m),
\end{equation}
where $N_{\textrm{SP}}(n)$ denotes the number of \cnots{} required for
state preparation on $n$ qubits starting from the state $\ket{0}_n$,
$N_{\Delta}(m)\leqslant2^m-2$ denotes the number of \cnots{} required
to decompose a diagonal gate acting on $m$ qubits~\cite{diag} and
$N_G(m,n)$ is the number of \cnots{} used to decompose the gates in
$\{G_{i}\}_{i \in \{1,2,\dots,2^m-1\}}$.
\end{thm} 
\begin{proof} 
  We decompose $V$ as described in Lemma~\ref{isocc}, and $\{G_i\}_{i
    \in \{1,2,\dots,2^m-1\}}$ as in the proof of
  Lemma~\ref{prop_iso}. By Lemma~\ref{prop_iso_CNOT_count} we have
\setlength{\arraycolsep}{0.0em}
\begin{eqnarray}
N_G(m,n)&{}={}&\sum_{k=1}^{2^m-1}2^n-n-1+ Q^k(n)N_{C_{n-1}(U)} \nonumber \\
&&{=}\; \left(2^m-1\right)\left(2^n-n-1\right)+Q(m,n)N_{C_{n-1}(U)} \nonumber 
\label{eqN_G(m,n)}
\end{eqnarray}
\setlength{\arraycolsep}{5pt} where
$Q(m,n)=2^m(n-\frac{m}{2}-1)-n+m+1$ is the number of MCGs used, as given
by Corollary~\ref{cor_total_number_mcg}, and $N_{C_{n-1}(U)}$ denotes
the number of \cnots{} needed to decompose a MCG ${C_{n-1}(U)}$, given
by Theorem~\ref{ThmC}. Note that we require $U\in SU(2)$ to
use Theorem~\ref{ThmC}. This causes no problems in our construction,
since Lemma~\ref{zero_entry} holds for $U \in SU(2)$. The gate
$G^{\dagger}_{0}$ can be decomposed using a decomposition scheme for
state preparation, which finishes the proof.\end{proof}

\begin{cor} [Explicit count for an isometry] \label{cor_isocc_CNOT_count} The number of \cnots{} required to
  decompose an $m$ to $n\geqslant 8$ isometry $V$ satisfies
  \setlength{\arraycolsep}{0.0em}
\begin{eqnarray}\label{eq:Iso_CNOT}
N_{\mathrm{iso}}(m,n)&{} \leqslant {}&\lceil 2^{m+n}-\frac{1}{24}2^n -2\cdot2^{\frac{n}{2}} \\
&&{+}\:2^m \left(28n^2+m(44-14n)-117n+88\right)\nonumber \\
&&{-}\:28n^2+m\left(28n-88\right)+117n-87 \rceil.\nonumber 
\end{eqnarray}
\setlength{\arraycolsep}{5pt}
\end{cor} 
\begin{proof} 
  Theorem~\ref{ThmC} implies that $N_{C_{n-1}(U)}\leqslant 28n-88$ for
  all $n$ (for simplicity we over-count in the case that $n$ is
  odd). The asymptotic best-known \cnot{} counts for state preparation
  (see Table~\ref{tab:Results}) give us the upper bound
  $N_{\textrm{SP}}(n) \leqslant
  \frac{23}{24}2^n-2\cdot2^{\frac{n}{2}}+2$. %, which holds for all $n$.
  The number of \cnots{} used to decompose a diagonal gate $\Delta$
  acting on $m$ qubits is at most
  $N_{\Delta}(m)=2^m-2$~\cite{diag}. Using the inequality~(\ref{eq:iso_CNOT_general_count}) this leads to the claimed count.\end{proof}

\subsection{Optimization of the decomposition of an isometry using the CSD} \label{sec:app_isocsd}

\begin{thm} [Optimized CSD approach] \label{opt_CNOT_count_CSD} 
Let $m$ and $n$ be natural numbers with $2\leqslant m\leqslant n$
  and $V$ be an isometry from $m$ qubits to $n$ qubits. There exists a
  decomposition of $V$ in terms of single-qubit gates and \cnots{}
  such that the number of \cnot{} gates required satisfies
\begin{equation}\label{eq:CSD_count_2}
  N_{\mathrm{iso}}(m,n)\leqslant  \frac{23}{144}\left(4^m+2\cdot4^n\right)-2^{m-1}-2^{n}+ \frac{1}{3}(m-n+4).
\end{equation}
\end{thm}

Note that we recover the optimized \cnot{} count for general quantum
gates~\cite{2} setting $n=m$ in the inequality~(\ref{eq:CSD_count_2}).
 
\begin{proof}  
We optimize the \cnot{} count of
  Section~\ref{sec:isocsd} using the two ideas described in the
  Appendix of~\cite{2}. There it is shown how one can combine the
  decomposition of the $C^{\textnormal{u}}_{i}(R_{y})$ gates with
  neighbouring $i$-qubit-$C^{\textnormal{u}}_{1}(U)$ gates to save one
  \cnot{} gate over what would be required if the
  $C^{\textnormal{u}}_{i}(R_{y})$ gates were decomposed on their
  own. The essential idea is to use the circuit identity
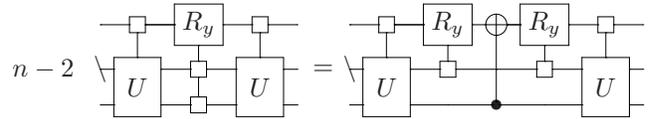
\[
\Qcircuit @C=0.5em @R=.45em {
&&&\gate{}\qwx[1]&\gate{R_y}\qwx[1]\qw&\gate{}\qwx[1]&\qw&&&&&\gate{}\qwx[1]&\gate{R_y}\qwx[1]\qw& \targ   &\gate{R_y}\qwx[1]\qw&\gate{}\qwx[1]&\qw\\
&\lstick{n-2}&{\backslash}&\multigate{1}{U}&\gate{}\qwx[1]&\multigate{1}{U}&\qw&&=&&{\backslash}&\multigate{1}{U}&\gate{}&\qw&\gate{}&\multigate{1}{U}&\qw\\
&&&\ghost{U}&\gate{}&\ghost{U}&\qw&&&&&\ghost{U}&\qw&\ctrl{-2}&\qw&\ghost{U}&\qw\\
}
\]

The same idea also works for the CSD adapted to isometries, allowing
us to save 1 \cnot{} per uniformly controlled $R_y$ gate.

To count the number of uniformly controlled $R_y$ gates
$Q_{R_{y}}(m,n)$ used for an $m$ to $n$ isometry using the
decomposition scheme of Section~\ref{sec:isocsd} we use the following
recursion relation:
\begin{align}
\label{eq:rec2}
		Q_{R_{y}}(m,i+1)&=Q_{R_{y}}(m,i)+
                \frac{2\cdot 4^{i-2}-2}{3}\!+\! 1 \text{ if }m\leqslant i < n\\
\label{eq:recs2}
		Q_{R_{y}}(m,m)&=\frac{4^{m-2}-1}{3},
\end{align}
where the last relation comes from Appendix~A of~\cite{2}.  Solving
these gives
\begin{equation}
\label{eq:rec2function}
		Q_{R_{y}}(m,n)=\frac{1}{144}\left(2^{2n+1}+4^m\right)+\frac{1}{3}\left(n-m-1\right).
\end{equation}

The CSD decomposition is used until the only generic unitaries that
remain are on two qubits.  In Appendix~B of~\cite{2} it is shown how
to save one \cnot{} gate for each of the remaining two-qubit gates
apart from one. Again this idea also works using the CSD adapted to
isometries. The number of two-qubit gates $Q_{U_{2}}(m,n)$ arising in
the decomposition scheme described in Section~\ref{sec:isocsd}
satisfies the following recursion relation:
\begin{align}
\label{eq:rec3}
		Q_{U_{2}}(m,i+1)&=Q_{U_{2}}(m,i)+2\cdot4^{i-2} \text{ if } m\leqslant i < n,\\
\label{eq:recs3}
		Q_{U_{2}}(m,m)&=4^{m-2},
\end{align}
where the last of these relations is taken from Appendix~B of~\cite{2}.
Solving these gives
\begin{equation}
\label{eq:rec3function}
		Q_{U_{2}}(m,n)=\frac{1}{48}\left(2^{2n+1}+4^m\right).
\end{equation}

The optimized \cnot{} count is thus given by
\begin{equation}
\label{eq:Nisocsd_general}
		N_{\mathrm{iso}}(m,n)=\tilde{N}_{\mathrm{iso}}(m,n)-Q_{R_{y}}(m,n)-Q_{U_{2}}(m,n)+1,
\end{equation}
where $\tilde{N}_{\mathrm{iso}}(m,n)$ is bounded by the inequality~(\ref{eq:CSD_count_1}). This leads to the claimed count.\end{proof}

\subsection{Optimized state preparation} \label{opt_SP_app}

For state preparation on two and three qubits there exist ad hoc
methods using one and three \cnot{} gates
respectively~\cite{SP_three_qubits1}.  For state preparation on
$n\geqslant 4$ qubits we use the decomposition scheme described
in~\cite{3}. In the case that $n$ is even, this uses the following
iterative circuit:
\[
\Qcircuit @C=0.45em @R=0.38em {
\lstick{\ket{0}}&&\qw&\multigate{8}{SP}&\qw&&&&&&&\lstick{\ket{0}}&&\qw&\qw&\multigate{3}{SP}&\qw&\ctrl{5}&\qw&\qw&\qw&\multigate{3}{U_1}&\qw\\
&&\vdots&&&&&&&&&&&\vdots&&&&&\ddots\\
&&&&&&&\\
\lstick{\ket{0}}&&\qw&\ghost{SP}&\qw&&&&&&&\lstick{\ket{0}}&&\qw&\qw&\ghost{SP}&\qw&\qw&\qw&\ctrl{5}&\qw&\ghost{U_1}&\qw\\
&&&&&&=&&&&&&&\\
\lstick{\ket{0}}&&\qw&\ghost{SP}&\qw&&&&&&&\lstick{\ket{0}}&&\qw&\qw&\qw&\qw&\targ&\qw&\qw&\qw&\multigate{3}{U_2}&\qw\\
&&\vdots&&&&&&&&&&&\vdots&&&&&\ddots\\
&&&&&&&\\
\lstick{\ket{0}}&&\qw&\ghost{SP}&\qw&&&&&&&\lstick{\ket{0}}&&\qw&\qw&\qw&\qw&\qw&\qw&\targ&\qw&\ghost{U_2}&\qw
}
\]
where we have divided the qubits into two groups of $n/2$.  In other
words, state preparation on $n$ qubits is equivalent to state
preparation on $n/2$ qubits, $n/2$ \cnots{}, and then two $n/2$ qubits
unitary operations. If $n$ is odd, the unitary $U_1$ is replaced by an $\lfloor n/2 \rfloor$-qubit unitary and $U_2$ by an
$\lfloor n/2\rfloor$ to $\lfloor n/2\rfloor+1$ isometry.

If $n$ is odd we can implement $U_2$ using the CSD approach.  Furthermore, we can use a similar technical trick as described in Appendix~B of~\cite{2} to save one \cnot{} gate when implementing $U_1$: as noted in Appendix~B of~\cite{2} all apart from one of the two-qubit gates arising in the decomposition of a general unitary can be decomposed using two \cnot{} gates.  For the last one we can also extract a diagonal gate and merge it with the state preparation, since the diagonal gate commutes through the control qubits of the \cnot{} gates that precede $U_1$.

In other words, for $n$ even, we have 
\setlength{\arraycolsep}{0.0em}
\begin{eqnarray}
N_{\textrm{SP}}(n)&{}\leqslant{}& N_{\textrm{SP}}\left(\frac{n}{2}\right)+\frac{n}{2}+2N_{\textrm{iso}}\left(\frac{n}{2},\frac{n}{2}\right)-1\nonumber \\
N_{\textrm{SP}}(n+1)&{}\leqslant{}&N_{\textrm{SP}}\left( \frac{n}{2}\right)+\frac{n}{2}+N_{\textrm{iso}} \left(\frac{n}{2},\frac{n}{2} \right)\nonumber\\
&&{+}\:N_{\textrm{iso}}\left(\frac{n}{2},\frac{n}{2}+1\right)-1\label{CNOT_SP_odd}\, ,
\end{eqnarray}%\setlength{\arraycolsep}{5pt}
where for the purpose of evaluating $N_{\textrm{iso}}$ in these
counts, we use the inequality~\eqref{eq:CSD_count_2}.  Starting from
$N_{\textrm{SP}}(2)=1$ and
$N_{\textrm{SP}}(3)=3$~\cite{SP_three_qubits1}, this allows us to
iteratively compute $N_{\textrm{SP}}(n)$ for increasing $n$.  For
illustration purposes, the circuit for state preparation on $4$ qubits
is shown in the following circuit diagram.
 
\[
\Qcircuit @C=0.45em @R=0.38em {
\lstick{\ket{0}}&\gate{U}&\ctrl{1}&\gate{U}&\ctrl{2}&\qw&\gate{U}&\ctrl{1}&\gate{U}&\ctrl{1}&\gate{U}&\qw&\qw&\qw\\
\lstick{\ket{0}}&\qw&\targ&\gate{U}&\qw&\ctrl{2}&\gate{U}&\targ&\gate{U}&\targ&\gate{U}&\qw&\qw&\qw\\
\lstick{\ket{0}}&\qw&\qw&\qw&\targ&\qw&\gate{U}&\targ&\gate{U}&\ctrl{1}&\qw&\targ&\gate{U}&\qw\\
\lstick{\ket{0}}&\qw&\qw&\qw&\qw&\targ&\gate{U}&\ctrl{-1}&\gate{U}&\targ&\gate{U}&\ctrl{-1}&\gate{U}&\qw\\\
}
\]

Note that the depth of the circuit is, to leading order, the number
of steps required to perform $U_2$, since $U_1$ and $U_2$ can be done
in parallel and dominate the gate count.

\section{Isometries on a small number of qubits} \label{small_cases_appendix}

\subsection{Isometries from one to two qubits} \label{1_2_iso_app}

We present an ad hoc decomposition for a $1$ to $2$ isometry $V$ reaching the theoretical lower bound of two \cnot{} gates. Our result is based  on the following decomposition of an  arbitrary two-qubit operator $U$ described in~\cite{2,unitary_lowerb1,unitary_lowerb2}.

\[
\Qcircuit @C=1.0em @R=.46em {
& \multigate{2}{U}&\qw&&&& \multigate{2}{\Delta}&\gate{U}&\ctrl{2} &\gate{R_y}&\ctrl{2} &\gate{U}&\qw \\
&  &        & =                            &&             &&&\\
&\ghost{U}  &\qw&&&& \ghost{\Delta}&  \gate{U}&\targ&\gate{R_y} &\targ&\gate{U}&\qw   \\
}
\]

We represent $V$ by a unitary matrix $V_2$ such that $V=V_2 I_{2^2\times2^1}$. Since we are only interested in the first two columns of $V_2$, we can replace the diagonal gate $\Delta$ of the last circuit by a single-qubit diagonal gate acting on the least significant qubit. Absorbing this gate into the neighbouring (arbitrary) single-qubit gate we conclude the following circuit equivalence.

\[
\Qcircuit @C=1.0em @R=.46em {
\lstick{\ket{0}}& \multigate{2}{V_2}&\qw&&&\lstick{\ket{0}}&\gate{U}&\ctrl{2} &\gate{R_y}&\ctrl{2} &\gate{U}&\qw \\
&  &        & =                            &&             &&&\\
&\ghost{V_2}  &\qw&&&&  \gate{U}&\targ&\gate{R_y} &\targ&\gate{U}&\qw   \\
}
\]

\subsection{Isometries leading to three
  qubit states}\label{app_sec_iso_on_three_qubits}

In this section we explain the steps needed to decompose isometries
from $m$ to $3$ qubits for $m=1$ and $m=2$. Note that for $m=0$ one
can use the decomposition scheme for state preparation given
in~\cite{SP_three_qubits1, SP_circuit_for_threee_qubits}, and for
$m=3$ the decomposition scheme of~\cite{2}.

\subsubsection{Isometries from one to three qubits} \label{app_sec_iso_one_to_three}

We use the column-by-column approach described in Section~\ref{sec:isounif} to decompose an isometry $V$ from one to three qubits. As in Section~\ref{sec_dec_schemes}, we represent the $8 \times 2$ matrix corresponding to $V$ by an $8 \times 8$ unitary matrix $G^{\dagger}$ by writing $V=G^{\dagger}I_{8 \times 2}$. The unitary $G_0^{\dagger}$ (defined in Section~\ref{sec:isounif}) corresponds to state preparation on three qubits ($G_0^{\dagger}\ket{0}^{\otimes 3}=V \ket{0}=:\ket{\psi_0^0}$) and can therefore be implemented with the techniques described in~\cite{SP_three_qubits1, SP_circuit_for_threee_qubits}. 

We now consider constructing a circuit for the unitary $G_1$. We
define $\ket{{\psi_1^0}}:=G_0V\ket{1}$ and note that its first entry
is zero. One can use Lemma~\ref{rotate_to_basis_state} to choose the
gates depicted in the circuit diagram below such that they have the
following action on $\ket{{\psi_1^0}}$ (as previously `$\ast$'
represents an arbitrary complex entry):
 \renewcommand{\arraystretch}{0.7}
\begin{equation*}
  \mathsmaller{\ket{{\psi_1^0}}} = \begin{array}{c@{\!\!\!}l}
  \left[ \begin{array}[c]{ccccc} 
     \mathsmaller{0}\\
     \mathsmaller{\ast}\\
    \mathsmaller{\ast}\\
      \mathsmaller{\ast}\\
     \mathsmaller{\ast}\\
     \mathsmaller{\ast}\\
    \mathsmaller{\ast}\\
      \mathsmaller{\ast}\\
  \end{array}  \right]
\end{array}
\phantom{a}
 \xrightarrow{ }  
 \begin{array}{c@{\!\!\!}l}
   \left[ \begin{array}[c]{ccccc} 
     \mathsmaller{0}\\
    \mathsmaller{\ast}\\
     \mathsmaller{0}\\
     \mathsmaller{\ast}\\
    \mathsmaller{0}\\
    \mathsmaller{\ast}\\
     \mathsmaller{0}\\
       \mathsmaller{\ast}\\
\end{array}  \right]
\end{array}
\phantom{a}
 \xrightarrow{ }  
 \begin{array}{c@{\!\!\!}l}
   \left[ \begin{array}[c]{ccccc} 
    \mathsmaller{0}\\
    \mathsmaller{\ast}\\
     \mathsmaller{0}\\
      \mathsmaller{0}\\
     \mathsmaller{0}\\
   \mathsmaller{\ast}\\
  \mathsmaller{0}\\
     \mathsmaller{\ast}\\
\end{array}  \right]
\end{array}
\phantom{a}
 \xrightarrow{}  
 \begin{array}{c@{\!\!\!}l}
   \left[ \begin{array}[c]{ccccc} 
      \mathsmaller{0}\\
    \mathsmaller{\ast}\\
     \mathsmaller{0}\\
      \mathsmaller{0}\\
     \mathsmaller{0}\\
   \mathsmaller{\ast}\\
  \mathsmaller{0}\\
     \mathsmaller{0}\\
\end{array}  \right]
\end{array}
\phantom{a}
 \xrightarrow{ }  
 \begin{array}{c@{\!\!\!}l}
   \left[ \begin{array}[c]{ccccc} 
     \mathsmaller{0}\\
    \mathsmaller{1}\\
   \mathsmaller{0}\\
     \mathsmaller{0}\\
    \mathsmaller{0}\\
 \mathsmaller{0}\\
     \mathsmaller{0}\\
    \mathsmaller{0}\\
\end{array}  \right]
\end{array}
\end{equation*}

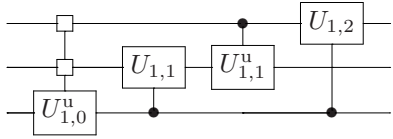
 \renewcommand{\arraystretch}{1}

\[
\Qcircuit @C=1em @R=.05em {
&&&&& \gate{} \qwx[1] \qw 							& \qw					&\ctrl{1} 						&\gate{U_{1,2}} \qw					 	 &\qw			\\
&&&&& \gate{} \qwx[1] \qw							&\gate{U_{1,1}} \qw			&\gate{U^{\textnormal{u}}_{1,1}}	& \qw 			&		\qw 		\\
&&&&& \gate{U^{\textnormal{u}}_{1,0}} \qw 				&\ctrl{-1} 					&\qw							&\ctrl{-2} \qw			&\qw 						 \\
}
\]

Note that all the gates in the circuit above act trivially on the
state $\ket{0}^{\otimes 3}$. Therefore this represents a valid circuit
for the unitary $G_1$.

\begin{rmk} The notation in the circuit diagram above is as introduced
  in the general case in Section~\ref{sec:isounif}. The difference
  between the circuit above and the circuit we would get by the
  techniques of Section~\ref{sec:isounif} is that we switch the order
  of the UCG and the MCG (note that they commute by construction) and
  leave away some controls of the MCGs. Indeed, similar
  simplifications are possible for the most MCG, which arise in the
  column-by-column decomposition of arbitrary isometries from $m$ to
  $n$ qubits. We have not taken this into account in the general
  \cnot{} count, since it does not affect its leading order.
\end{rmk}

Since MCGs are a special case of UCGs, we can implement the MCGs using
UCGs instead. Furthermore, we can implement all the UCGs up to
diagonal gates (i.e., implement $\Delta C$ rather than $C$ for each
UCG $C$) and correct for these at the end using a diagonal gate
applied to the least significant qubit. Doing so we can save some
\cnots, because for small $n$, we know how to implement $\Delta
C^{\textrm{u}}_{n-1}(U)$ more efficiently than $C_{n-1}(U)$. For
example, we need $8$ \cnot{} gates to implement a $C_{2,3}(U)$ gate
(cf.\ Lemma~\ref{cor5.3} and~\ref{lem6.1}) and only $3$ \cnot{} gates
to implement a $\Delta C^{\textrm{u}}_2(U)$ gate
(cf.\ Table~\ref{tab:counts_overview}).

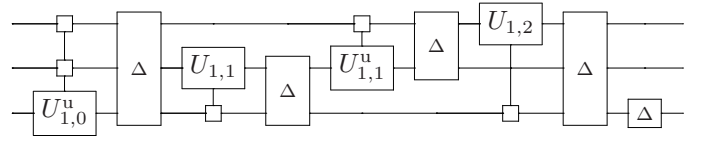
\[
\Qcircuit @C=0.8em @R=.05em {
& \gate{} \qwx[1] \qw 	&	\multigate{2}{\mathsmaller{\Delta}}						& \qw	& \qw				& \gate{} \qwx[1] \qw 		&	\multigate{1}{\mathsmaller{\Delta}}					&\gate{U_{1,2}} \qw		&	\multigate{2}{\mathsmaller{\Delta}}		 & \qw	 &\qw			\\
& \gate{} \qwx[1] \qw	& \ghost{\mathsmaller{\Delta}}								&\gate{U_{1,1}} \qw	&	\multigate{1}{\mathsmaller{\Delta}}	&\gate{U^{\textnormal{u}}_{1,1}}& \ghost{\mathsmaller{\Delta}}				& \qw 	& \ghost{\mathsmaller{\Delta}}		&		\qw 	 &\qw	\\
& \gate{U^{\textnormal{u}}_{1,0}} \qw 	& \ghost{\mathsmaller{\Delta}}						& \gate{} \qwx[-1] \qw 	 	& \ghost{\mathsmaller{\Delta}}		&\qw				&\qw							& \gate{} \qwx[-2] \qw  & \ghost{\mathsmaller{\Delta}}		&\gate{\mathsmaller{\Delta}}		 &\qw					 \\
}
\]

We implement each UCG together with its subsequent diagonal gate as
described in~\cite{10}. Together with the circuit for the unitary
$G_0$, this leads to the following circuit for the isometry $V$

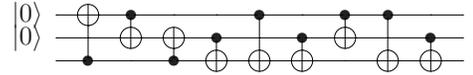
\[
\Qcircuit @C=0.8em @R=.05em {
 \lstick{\ket{0}}&\targ&\ctrl{1}&\qw&\qw&\ctrl{2}&\qw&\ctrl{1}&\ctrl{2}&\qw  &\qw	\\
 \lstick{\ket{0}}&\qw&\targ&\targ&\ctrl{1}&\qw&\ctrl{1}&\targ&\qw&\ctrl{1}  &\qw	\\
& \ctrl{-2}&\qw& \ctrl{-1}&\targ&\targ&\targ&\qw&\targ&\targ   &\qw	\\
}
\]
where we have not depicted the single-qubit gates for simplicity.

\begin{figure*}[!t] 
\centering
 \renewcommand{\arraystretch}{0.5}

\begin{equation*}
  \mathsmaller{\ket{{\psi_1^0}}} = \begin{array}{c@{\!\!\!}l}
  \left[ \begin{array}[c]{ccccc} 
     \mathsmaller{0}\\
     \mathsmaller{\ast}\\
  \hline
    \mathsmaller{\ast}\\
      \mathsmaller{\ast}\\
    \hline
     \mathsmaller{\ast}\\
     \mathsmaller{\ast}\\
    \hline
    \mathsmaller{\ast}\\
      \mathsmaller{\ast}\\
    \hline
     \mathsmaller{\ast}\\
     \mathsmaller{\ast}\\
    \hline
    \mathsmaller{\ast} \\
      \mathsmaller{\ast}\\
    \hline
     \mathsmaller{\ast}\\
     \mathsmaller{\ast}\\
    \hline
    \mathsmaller{\ast}\\
      \mathsmaller{\ast}\\
  \end{array}  \right]
\end{array}
\phantom{a}
 \xrightarrow{ C_{3}^{\textnormal{u}}(U_{1,0}^{\textnormal{u}}) }  
 \begin{array}{c@{\!\!\!}l}
   \left[ \begin{array}[c]{ccccc} 
     \mathsmaller{0}\\
    \mathsmaller{\ast}\\
     \mathsmaller{0}\\
     \mathsmaller{\ast}\\
    \hline
    \mathsmaller{0}\\
    \mathsmaller{\ast}\\
     \mathsmaller{0}\\
       \mathsmaller{\ast}\\
    \hline
    \mathsmaller{0}\\
    \mathsmaller{\ast}\\
     \mathsmaller{0}\\
          \mathsmaller{\ast}\\
    \hline
    \mathsmaller{0}\\
     \mathsmaller{\ast}\\
    \mathsmaller{0}\\
    \mathsmaller{\ast}\\
\end{array}  \right]
\end{array}
\phantom{a}
 \xrightarrow{ C_{1}(U_{1,1})}  
 \begin{array}{c@{\!\!\!}l}
   \left[ \begin{array}[c]{ccccc} 
    \mathsmaller{0}\\
    \mathsmaller{\ast}\\
     \mathsmaller{0}\\
      \mathsmaller{0}\\
      \hline
     \mathsmaller{0}\\
   \mathsmaller{\ast}\\
  \mathsmaller{0}\\
     \mathsmaller{\ast}\\
    \hline
     \mathsmaller{0}\\
    \mathsmaller{\ast}\\
   \mathsmaller{0}\\
     \mathsmaller{\ast}\\
      \hline
     \mathsmaller{0}\\
    \mathsmaller{\ast}\\
   \mathsmaller{0}\\
     \mathsmaller{\ast}\\
\end{array}  \right]
\end{array}
\phantom{a}
 \xrightarrow{C^{\textnormal{u}}_{2}(U_{1,1}^{\textnormal{u}})}  
 \begin{array}{c@{\!\!\!}l}
   \left[ \begin{array}[c]{ccccc} 
      \mathsmaller{0}\\
    \mathsmaller{\ast}\\
     \mathsmaller{0}\\
      \mathsmaller{0}\\
     \mathsmaller{0}\\
   \mathsmaller{\ast}\\
  \mathsmaller{0}\\
     \mathsmaller{0}\\
    \hline
     \mathsmaller{0}\\
    \mathsmaller{\ast}\\
   \mathsmaller{0}\\
     \mathsmaller{0}\\
     \mathsmaller{0}\\
    \mathsmaller{\ast}\\
   \mathsmaller{0}\\
     \mathsmaller{0}\\
\end{array}  \right]
\end{array}
\phantom{a}
 \xrightarrow{C_{1}(U_{1,2})  }  
 \begin{array}{c@{\!\!\!}l}
   \left[ \begin{array}[c]{ccccc} 
     \mathsmaller{0}\\
    \mathsmaller{\ast}\\
   \mathsmaller{0}\\
     \mathsmaller{0}\\
    \mathsmaller{0}\\
 \mathsmaller{0}\\
     \mathsmaller{0}\\
    \mathsmaller{0}\\
    \hline
         \mathsmaller{0}\\
    \mathsmaller{\ast}\\
   \mathsmaller{0}\\
     \mathsmaller{0}\\
    \mathsmaller{0}\\
 \mathsmaller{\ast}\\
     \mathsmaller{0}\\
    \mathsmaller{0}\\
\end{array}  \right]
\end{array}
\phantom{a}
 \xrightarrow{C_{1}(U_{1,2}^{\textnormal{u}}) }  
 \begin{array}{c@{\!\!\!}l}
   \left[ \begin{array}[c]{ccccc} 
      \mathsmaller{0}\\
    \mathsmaller{\ast}\\
   \mathsmaller{0}\\
     \mathsmaller{0}\\
    \mathsmaller{0}\\
 \mathsmaller{0}\\
     \mathsmaller{0}\\
    \mathsmaller{0}\\
         \mathsmaller{0}\\
    \mathsmaller{\ast}\\
   \mathsmaller{0}\\
     \mathsmaller{0}\\
    \mathsmaller{0}\\
 \mathsmaller{0}\\
     \mathsmaller{0}\\
    \mathsmaller{0}\\

\end{array}  \right]
\end{array}
\phantom{a}
 \xrightarrow{C_{1}(U_{1,3}) }  
\begin{array}{c@{\!\!\!}l}
  \left[ \begin{array}[c]{ccccc} 
      \mathsmaller{0}\\
   \mathsmaller{1}\\
     \mathsmaller{0}\\
    \mathsmaller{0}\\
 \mathsmaller{0}\\
     \mathsmaller{0}\\
    \mathsmaller{0}\\
       \mathsmaller{0}\\
    \mathsmaller{0}\\
   \mathsmaller{0}\\
     \mathsmaller{0}\\
    \mathsmaller{0}\\
 \mathsmaller{0}\\
     \mathsmaller{0}\\
    \mathsmaller{0}\\
        \mathsmaller{0}\\
         \end{array}  \right]
  \end{array}
\end{equation*}

\[
\Qcircuit @C=4.em @R=.05em {
&& \gate{} \qwx[1] \qw 							& \qw					& \gate{} \qwx[1] \qw 		 		& \qw	&\ctrl{1} \qw			&\gate{U_{1,3}}	 	 &\qw			\\
&& \gate{} \qwx[1] \qw 							& \qw					& \gate{} \qwx[1] \qw  			&\gate{U_{1,2}}&\gate{U^{\textnormal{u}}_{1,2}} \qw	&		\qw				 	 &\qw			\\
&& \gate{} \qwx[1] \qw							&\gate{U_{1,1}} \qw			&\gate{U^{\textnormal{u}}_{1,1}}	& \qw 			&		\qw&		\qw&		\qw 		\\
&& \gate{U^{\textnormal{u}}_{1,0}} \qw 				&\ctrl{-1} 					&\qw							&\ctrl{-2} \qw		&		\qw&\ctrl{-3} \qw	&\qw 						 \\
}
\]

\caption{Implementing the second column of an isometry $V$ from
  one to four qubits with optimized controlling of the MCGs. Note that all gates act trivially on
  $\ket{0000}$. The symbol ``$\ast$" denotes an arbitrary complex
  number.}
\label{fig:coloumn2_optimized_MCG}
\end{figure*}

 \renewcommand{\arraystretch}{1}

\subsubsection{Isometries from two to three qubits}\label{app:two_to_three}

We use the CSD-approach described in Section~\ref{sec:isocsd} to decompose an isometry, $V$, from two to three qubits. As in Section~\ref{sec_dec_schemes}, we represent the $8 \times4$ matrix corresponding to $V$ by an $8 \times 8$ unitary matrix
$G^{\dagger}$, by writing $V=G^{\dagger} I_{8 \times 4}$. Then we apply Theorem~10 of~\cite{2} to $G^{\dagger}$, which gives us
\[
\Qcircuit @C=1em @R=.2em {
&\lstick{\ket{0}}&			& \multigate{2}{G^{\dagger}}	&\qw& &&	&\lstick{\ket{0}}& \gate{}\qwx[2]       & \gate{R_{y}} 	&  \gate{}  \qwx[2]	&   \qw	\\
&& 			&	&	&=			&&&       	     & &	 			   & &   & 	  \\
&\lstick{2}& {\backslash}	& \ghost{G^{\dagger}}&\qw &&&& {\backslash}	&\gate{A}   &\gate{} \qwx[-2] \qw&		\gate{B} 	&   \qw	 \\
}
\]
where each of the symbols $A$ and $B$ is a placeholder for two two-qubit unitaries denoted by $\{A_0, A_1\}$ or $\{B_0, B_1\}$ respectively. Since we can assume that the first qubit is initially in the state $\ket{0}$, we always implement $A_0$ on the last two qubits at the start of the circuit (on the right hand side) above. Therefore we can simplify the above circuit.

\[
\Qcircuit @C=1em @R=.2em {
&\lstick{\ket{0}}&			& \multigate{2}{G^{\dagger}}	&\qw& &	&&\lstick{\ket{0}}& \qw      & \gate{R_{y}} 	&  \gate{}  \qwx[2]	&   \qw	\\
&& 			&	&	&=			&&&&       	      &	 			   & &   & 	  \\
&\lstick{2}& {\backslash}	& \ghost{G^{\dagger}}&\qw &&&& {\backslash}	&\gate{A_0}   &\gate{} \qwx[-2] \qw&		\gate{B} 	&   \qw	 \\
}
\]

We apply Theorem~8 of~\cite{2} to the uniformly controlled $R_y$ gate. Together with Appendix A of~\cite{2} this leads to the following circuit for the isometry $V$
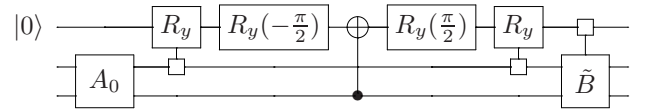
\[
\Qcircuit @C=0.7em @R=.2em {
\lstick{\ket{0}}	& \qw   &\gate{R_{y}}  \qw      &\gate{R_{y}(-\frac{\pi}{2})}  \qw	        &\targ	 &\gate{R_{y}(\frac{\pi}{2})}   \qw	 &\gate{R_{y}}  \qw	&\gate{} \qwx[1] \qw &   \qw	\\
&\multigate{1}{A_0} &\gate{} \qwx[-1] \qw& \qw&     \qw	&  \qw &\gate{} \qwx[-1] \qw & \multigate{1}{\tilde{B}}	 &  \qw  \\
 	&\ghost{A_0}&\qw	 	&   \qw		 &\ctrl{-2}		&  \qw 	&	  \qw &\ghost{\tilde{B}} &	  \qw	 \\
}
\]
where we can absorb the $R_{y}(\frac{\pi}{2})$ and $R_{y}(-\frac{\pi}{2})$ gates into the neighbouring uniformly controlled $R_y$ gates. We apply Theorem~12 of~\cite{2} to the last uniformly controlled gate in the circuit above, which gives us two two-qubit unitaries $U$ and $W$ and the following circuit for the isometry $V$.

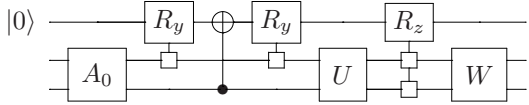
\[
\Qcircuit @C=0.7em @R=.2em {
\lstick{\ket{0}}	& \qw   &\gate{R_{y}}  \qw   	        &\targ	 	 &\gate{R_{y}}  \qw	& \qw &\gate{R_{z}}		&   \qw	&   \qw	\\
&\multigate{1}{A_0} &\gate{} \qwx[-1] \qw&      \qw	&\gate{} \qwx[-1] \qw & \multigate{1}{U}&\gate{} \qwx[-1] \qw	& \multigate{1}{W}	 &  \qw  \\
&\ghost{A_0}&\qw	 			 &\ctrl{-2}			&	  \qw &\ghost{U} 	&\gate{} \qwx[-1] \qw	&\ghost{W} 	&	  \qw	 \\
}
\]
\\

Decomposing the uniformly controlled rotations as described in~\cite{2} and using the techniques described in Appendix B of~\cite{2} leads to the following circuit for $V$

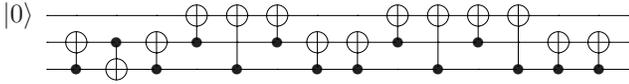
\[
\Qcircuit @C=0.7em @R=.2em {
\lstick{\ket{0}}		&  \qw &  \qw		&  \qw  	&\targ 	&\targ&\targ&  \qw &  \qw	&\targ 	&\targ&\targ&\targ&  \qw&  \qw&  \qw		 \\
				&\targ	 &\ctrl{1}	&\targ	&\ctrl{-1}	&  \qw&\ctrl{-1}&\targ&\targ&\ctrl{-1}	&  \qw&\ctrl{-1}&\qw&\targ&\targ&  \qw	 \\
				 &\ctrl{-1}	&\targ	&\ctrl{-1}	&  \qw	&\ctrl{-2}&  \qw&\ctrl{-1}&\ctrl{-1}&  \qw	&\ctrl{-2}&  \qw&\ctrl{-2}&\ctrl{-1}&\ctrl{-1}&  \qw	 \\
}
\]
where the single-qubit gates are not depicted for simplicity.

\subsection{Isometries leading to four qubit states}

In this section we explain the steps needed to decompose isometries
from $m$ to $4$ qubits for $m=1$ and $m=2$. Note that for $m=0$ one
can use the decomposition scheme for state preparation described in Appendix~\ref{opt_SP_app}, and for
$m=4$  the decomposition scheme of~\cite{2}. The case $m=3$ can be done with the CSD-approach requiring $73$ \cnots{} (cf.\ equation~(\ref{eq:CSD_count_2}), and Appendix~\ref{app:two_to_three} for an example using the CSD-approach).  

\subsubsection{Isometries from one to four qubits}  \label{app_sec_iso_one_to_four}

As in Section~\ref{sec_dec_schemes}, we represent the $16 \times 2$ matrix corresponding to $V$ by an $16 \times 16$ unitary matrix $G^{\dagger}$ by writing $V=G^{\dagger}I_{16 \times 2}$. The unitary $G_0^{\dagger}$ (defined in Section~\ref{sec:isounif}) corresponds to state preparation on four qubits ($G_0^{\dagger}\ket{0}^{\otimes 4}=V \ket{0}=:\ket{\psi_0^0}$) and can therefore be implemented with the techniques described in Appendix~\ref{opt_SP_app} with $8$ \cnots{}. We construct the unitary $G_1$ in a similar fashion as in the case of a one to three isometry (cf.\ Appendix~\ref{app_sec_iso_one_to_three}) using the  column-by-column approach described in Section~\ref{sec:isounif}. This leads to a circuit for the unitary $G_1$ given in Fig.~\ref{fig:coloumn2_optimized_MCG}. We implement all MCG of the circuit for $G_1$ with UCG up to diagonal gates by the techniques described in~\cite{10} and correct for this at the end of the circuit with a diagonal gate acting on the least significant qubit (cf.\ Section~\ref{app_sec_iso_one_to_three}). Therefore we use $22$ \cnots{} to implement an isometry from $1$ to $4$ qubits.

\subsubsection{Isometries from two to four qubits}

As in Section~\ref{sec_dec_schemes}, we represent the $16 \times 4$ matrix corresponding to $V$ by a $16 \times 16$ unitary matrix $G^{\dagger}$ by writing $V=G^{\dagger}I_{16 \times 4}$. We can construct the unitaries $G_0$ and $G_1$ as described in Appendix~\ref{app_sec_iso_one_to_four}. Similary we find the following circuit for the unitary $G_2$ 

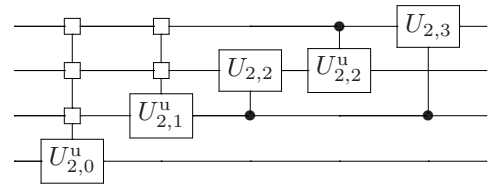
\[
\Qcircuit @C=1em @R=.05em {
&&&&& \gate{} \qwx[1] \qw 										& \gate{} \qwx[1] \qw 		 		& \qw			&\ctrl{1} \qw			&\gate{U_{2,3}}	 	 &\qw			\\
&&&&& \gate{} \qwx[1] \qw 										& \gate{} \qwx[1] \qw  			&\gate{U_{2,2}}		&\gate{U^{\textnormal{u}}_{2,2}} \qw	&		\qw				 	 &\qw			\\
&&&&& \gate{} \qwx[1] \qw										&\gate{U^{\textnormal{u}}_{2,1}}	& \ctrl{-1}\qw 		&		\qw&	\ctrl{-2}	\qw&		\qw 		\\
&&&&& \gate{U^{\textnormal{u}}_{2,0}} \qw 				 			&\qw							& \qw			&		\qw& \qw	&\qw 						 \\
}
\]

and the following circuit for the unitary $G_3$.

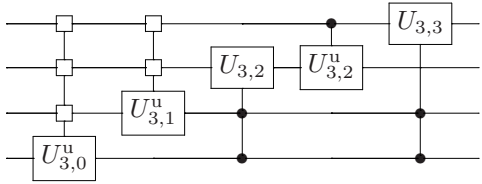
\[
\Qcircuit @C=1em @R=.05em {
&&&&& \gate{} \qwx[1] \qw 										& \gate{} \qwx[1] \qw 		 		& \qw			&\ctrl{1} \qw			&\gate{U_{3,3}}	 	 &\qw			\\
&&&&& \gate{} \qwx[1] \qw 										& \gate{} \qwx[1] \qw  			&\gate{U_{3,2}}		&\gate{U^{\textnormal{u}}_{3,2}} \qw	&		\qw				 	 &\qw			\\
&&&&& \gate{} \qwx[1] \qw										&\gate{U^{\textnormal{u}}_{3,1}}	& \ctrl{-1}\qw 		&		\qw&	\ctrl{-2}	\qw&		\qw 		\\
&&&&& \gate{U^{\textnormal{u}}_{3,0}} \qw 				 			&\qw							& \ctrl{-1}\qw			&		\qw&\ctrl{-1} \qw	&\qw 						 \\
}
\]

Note that two controls are required for the MCG for the unitary $G_3$, such that $G_3$ acts trivially on the states $\ket{0000}$, $\ket{0001}$ and $\ket{0010}$. 

We implement all MCG with UCG up to diagonal gates by the techniques described in~\cite{10} and correct for this at the end of the circuit with a diagonal gate acting on the two least significant qubits. Since a diagonal gate on two qubits requires $2$ \cnot{} gates~\cite{diag}, we conclude that we need $54$ \cnots{} to implement a two to four isometry.

\end{document}